\documentclass[aps,twocolumn,pra,10pt,superscriptaddress,showpacs]{revtex4}
\usepackage[bbgreekl]{mathbbol}
\usepackage{amsmath}
\usepackage{amssymb}
\usepackage{stmaryrd}
\usepackage{bbm}
\usepackage{graphicx}
\usepackage{framed}
\usepackage{color}
\usepackage{hyperref}
\usepackage{epstopdf}
\usepackage{dsfont}
\usepackage{amssymb}
\usepackage{amsmath}
\usepackage{amsthm}
\usepackage{wrapfig}
\usepackage{relsize}
\usepackage{bm}
\usepackage{comment}
\usepackage{float}

\newfloat{footnote}{hb}

\usepackage{ulem}

\usepackage{enumitem}
\usepackage{natbib}
\usepackage{graphicx}
\usepackage{physics}
\usepackage{hyperref}
\usepackage{xcolor}
\usepackage{subfigure}

\newtheorem{theorem}{Property}

\newtheorem{corollary}{Corollary}[theorem]

\begin{document}

	\title{Low-ground/High ground capacity regions analysis for Bosonic Gaussian Channels}
		\author{Farzad Kianvash}	\email{farzad.kianvash@sns.it}

		\affiliation{NEST, Scuola Normale Superiore and Istituto Nanoscienze-CNR, I-56126 Pisa, Italy}
	\author{Marco Fanizza}
	\affiliation{{F\'{\i}sica Te\`{o}rica: Informaci\'{o} i Fen\`{o}mens Qu\`{a}ntics, Departament de F\'{i}sica, Universitat Aut\`{o}noma de Barcelona, ES-08193 Bellaterra (Barcelona), Spain}. }
	\affiliation{NEST, Scuola Normale Superiore and Istituto Nanoscienze-CNR, I-56126 Pisa, Italy}
	\author{Vittorio Giovannetti}	\affiliation{NEST, Scuola Normale Superiore and Istituto Nanoscienze-CNR, I-56126 Pisa, Italy}
	\date{\today}

	\begin{abstract}
We present a comprehensive characterization of the interconnections between single-mode, phase-insensitive Gaussian Bosonic Channels resulting from channel concatenation. This characterization enables us to  identify, in  the parameter space of these maps, two distinct regions: low-ground and high-ground. In the low-ground region, the information capacities are smaller than a designated reference value, while in the high-ground region, they are provably greater. As a direct consequence, we systematically outline an explicit set  of upper bounds for the quantum and private capacity of these maps, which combine known upper bounds and composition rules,  improving upon existing results.	\end{abstract}
	
	\maketitle

\section{Introduction}
The efficiency of classical communication lines can be expressed using a single, simple formula~\cite{shannon1,shannon2}. However, when it comes to quantum communication lines (quantum channels) that utilize quantum systems as information carriers instead of classical signals~\cite{HOL BOOK, WILDE BOOK,VGHOL,BENSHOR}, this simplification no longer holds. Instead, a multitude of different and computationally challenging capacity functionals are required to fully assess the quality of these transmission lines. For instance, the classical capacity  of a quantum channel,  characterizes  the optimal rate of classical bits that can be reliably transferred per channel uses; the 
quantum capacity instead provides  the optimal rate of transmitted qubits, and the
private capacity,  the optimal rate of bits that can be transmitted privately along the channel.
 In our study, we specifically focus on a special class of quantum communication lines known as Gaussian Bosonic Channels (GBCs), which are commonly employed to model communication procedures utilizing the electromagnetic field as the carrier of transmitted messages~\cite{serafini,CAVES,HOLEVO01,ADVS}.
Despite significant progress made in recent years, the analysis of GBCs still presents complex challenges. Specifically, computing the exact values of certain information capacities for these maps requires optimization techniques that remain difficult to tackle. In principle, calculating these quantities necessitates taking the limit of properly regularized entropic functionals, considering the potential utilization of entanglement across multiple channel uses~\cite{S,Lioyd,Q CAP DEV,privCWY,privDS,shor superadditivity DC, Di vincenzo superadditivity DC, smith superadditivity, Fern superadditivity, q superadditivity, gaussian q superadditivity, Cubit superadditivity, dephrasure, bl superadditivity,c superadditivity,p superadditivity,tradeoff1,tradeoff2, vikesh}.
Given these complexities, the derivation of upper and lower bounds for the capacities of significant channels represents crucial progress in the field.

An established strategy for upper bounding information capacities involves utilizing data processing inequalities. For instance, it is possible to obtain an upper bound for the capacity of a specific channel by expressing it as a concatenation of channels whose capacities are already known or upper bounded \cite{SS UPP B DEP,Ouyang,SUTT UPP B DEP, lownoiseQ, distdepo,matteo, swat, flagged channel 1,flagged channel 2,wang,flagged channel 3,naj}. In this article, we employ this method to enhance the previously established bounds in the literature \cite{matteo,Pirandola upp bound,naj,flagged channel 3} for the quantum and private capacity of single-mode, phase-insensitive Gaussian Bosonic Channels (PI-GBCs). To achieve this result, we present a detailed decomposition of the parameter space of PI-GBC maps into  regions encompassing all channels that can simulate a given channel through concatenation with other PI-GBC elements.

The structure of the article is as follows: 
In Sec.~\ref{prel} we introduce the fundamental concepts and notation used in this article. We start by giving an overview of continuous variable quantum channels and establish crucial notation. Then, we briefly discuss various quantum capacities of a quantum channel, namely quantum capacity, private capacity, two-way quantum capacity, and secret-key capacity. Following that, we introduce phase-insensitive one-mode GBCs and establish specific notation to aid us throughout the manuscript. In Sec.~\ref{rev q cap} we provide a concise review of the current state-of-the-art bounds for the quantum capacity of GBCs. We discuss the main techniques used to derive these bounds, including the use of data processing inequalities and channel concatenation. We also highlight the key challenges that remain in computing the exact quantum capacity for these channels.
In Sec.~\ref{l and h reg}  we study the parameter space of PI-GBCs in terms of channel concatenation. We present a detailed decomposition of the parameter space of PI-GBC maps into regions encompassing all channels that can simulate a given channel through concatenation with other PI-GBC elements. This analysis allows us to identify the channels that can be used to upper bound the quantum capacity of a given PI-GBC.
In Sec.~\ref{sec:stab} we derive new upper bounds for the quantum and private capacity of single-mode, phase-insensitive Gaussian Bosonic Channels (PI-GBCs) by employing the channel concatenation method discussed in Sec.~\ref{l and h reg}. We demonstrate that our new bounds improve upon the previously established bounds in the literature, providing a more accurate estimation of the capacities for these channels.
Sec.~\ref{conc} concludes the manuscript.

\section{Preliminaries}\label{prel}
A quantum communication line connecting two distant parties
can be seen as a physical transformation that associates  the states of a system $A$, representing the input messages of the model,  with the  states of a 
second system $B$, representing the associated output messages. 
At mathematical level such an object is described as a completely positive trace preserving (LCPTP) linear map 
$\Lambda: \mathcal{B}_1(\mathcal{H}_A) \mapsto \mathcal{B}_1(\mathcal{H}_B)$ that links the set of the  trace-class  operators
of the  (possibly infinite dimensional) Hilbert spaces $\mathcal{H}_A$, $\mathcal{H}_B$ associated with 
$A$ and $B$ respectively~\cite{HOL BOOK, WILDE BOOK}.
By Stinespring representation we can always express $\Lambda$ as a reduction of an isometry $\hat{V}$ that connects ${\cal H}_A$ to an extension 
${\cal H}_{BE}$ of $\mathcal{H}_B$, i.e. 
$\Lambda(\cdots)=\Tr_E (\hat V \cdots {\hat V}^\dagger)\,$, with $\Tr_E$ being the partial trace with respect to $E$. 
 Such a construction allows us to introduce the notion of complementary channel $\tilde{\Lambda}:\mathcal{B}_1(\mathcal{H}_A) \mapsto \mathcal{B}_1(\mathcal{H}_E)$ defined as  $\tilde{\Lambda}(\cdots):=\Tr_A (\hat V\cdots {\hat V}^\dagger)$, which can be interpreted as the transformation 
 induced on the environment of the communication line by the signaling process~\cite{degradability}. 
 
 Similarly to what happens in classical information theory, the efficiency of a quantum channel $\Lambda$ can be gauged in terms of 
 a series figures of merit (the quantum capacities of the channel) that evaluate the optimal ratio between the amount of data which can be sent reliably through the channel and the total amount of redundancy needed to achieve such a goal~\cite{HOL BOOK, WILDE BOOK,VGHOL,BENSHOR}. 
In this paper we focus on special instances  of such quantities which in the context of continuous variable quantum information processing (see next section), admit optimal finite values even when allowing unbounded energy resources, i.e. 
the quantum capacity $Q(\Lambda)$, the private capacity $P(\Lambda)$, the two-way quantum capacity $Q_2(\Lambda)$, and the secret-key  capacity $K(\Lambda)$~\cite{WILDE BOOK}. They are hierarchically ordered via the inequalities
\begin{eqnarray}K(\Lambda) \geq Q_2(\Lambda),P(\Lambda)\geq  Q(\Lambda)\;. 
\end{eqnarray} 
The smallest among such terms, i.e. $Q(\Lambda)$, measures the maximum rate at which the communication line can transmit quantum information reliably over asymptotically many uses of the channel~\cite{S,Lioyd,Q CAP DEV}; $P(\Lambda)$ is instead the maximum rate at which we can transmit classical messages through the channel $\Lambda$ 
in such a way that an external party that is monitoring the line, will not be able to read such messages~\cite{Q CAP DEV};
$Q_2(\Lambda)$  represents the maximum quantum information transmission  rate attainable by allowing the communicating party to
use (arbitrary) distillation protocols through the use of a classical side-channel~\cite{BENNETT1}; and finally  the largest  of these quantities, i.e. $K(\Lambda)$,
 describes the maximum rate at which two parties can use the channel to distill secret random string of bits. 

 Despite being operationally well defined, no universal  formula is known that allows one to explicitly compute the values  of $Q_2(\Lambda)$ and $K(\Lambda)$ as entropic functionals. On the contrary,  such characterizations  exist for $Q(\Lambda)$ and for $P(\Lambda)$,
 based on regularized optimizations of , respectively, the output coherent information for $Q$, and the Holevo information gap between $\Lambda$ and its complementary map $\tilde{\Lambda}$, for $P$.
 Even in these cases, however, the explicit computation of $Q(\Lambda)$ and $P(\Lambda)$  is typically rather challenging and has been carried out only a very limited set of models (in particular for the special classes of degradable and anti-degradable maps). 
 A possible way to 
 circumvent this problem is to make use of data-processing inequalities. 
Specifically a simple resource counting argument can be invoked to observe that, if a quantum channel $\Lambda$ can be expressed
in terms of a LCPTP map $\Lambda'$ via the  
 concatenated action of other two LCPTP linear maps 
 $\Lambda_1$, $\Lambda_2$, then the following relations hold 
 \begin{eqnarray} \label{data} 
\Lambda = \Lambda_2 \circ\Lambda' \circ  \Lambda_1 \quad \Longrightarrow \quad 
{\cal K}(\Lambda)\leq  {\cal K}(\Lambda') \;,
 \end{eqnarray} 
 where hereafter we shall use the symbol ${\cal K}$ to represent an arbitrary capacity (e.g. $Q$,  $P$, $Q_2$, or $K$)~see e.g.~\cite{HOL BOOK, WILDE BOOK,VGHOL}.
Accordingly if the capacity values of $\Lambda_1$ or $\Lambda_2$ are known, or if upper bounds for those quantities are available,  
we can then use (\ref{data}) to constraint 
 the performances of $\Lambda$. Alternatively, if instead the capacity of $\Lambda$ is known or if a lower bound for it
 is available, we can use
(\ref{data}) to provide lower bounds for those of $\Lambda_1$ and $\Lambda_2$. 
In what follows, we shall make use of this simple idea to improve the capacity analysis of a special class of quantum maps which plays an important role in quantum information theory, that is the Bosonic Gaussian Channels  set, whose properties are briefly reviewed in the next subsection.

\subsection{Bosonic Gaussian Channels} \label{sec:BGC} 
Bosonic Gaussian Channels (BGCs) model a vast collection of noise models that tamper
communication schemes which rely on  the uses of e.m. signals~\cite{serafini,VGHOL}.
Formally they can be introduced as a special set of LCPT transformations which act on the Hilbert space $L^2(\mathbb R^n)$ of the square integrable functions, representing the states of $n$ independent harmonic oscillators each corresponding to an individual mode of the field. 
Indicating with $\mathbf{\hat{r}}:=(\hat{x}_1,\hat{p}_1,...,\hat{x}_n,\hat{p}_n)^\text{T}$  the set of canonical position and momentum operators of the modes,
the action of a BGC map  can  be assigned in terms of linear mappings
they induce on the  first and second canonical momenta of the quantum states  $\hat{\rho}\in \mathfrak{S}(L^2(\mathbb R^n))$ of the model, i.e. the $2n$-dimensional real vector $\mathbf{m}:=\Tr(\mathbf{\hat{r}}\hat\rho)$
and the $2n\times 2n$ real matrix $V:=\Tr(\{\mathbf{(\hat{r}-m),(\hat{r}-m)}^\text{T}\}\hat\rho)$. 

For what it concerns the present work we shall limit the analysis to  a special subset of single-mode ($n=1$) Phase Insensitive  GBCs (or PI-GBCs in brief) formed by the
maps 
${\Phi}_{x,M}$ 
characterized by two positive noise parameters $x,M\geq 0$, whose 
action on  the system is fully determined by the transformations
 \begin{align}\label{defgenerale}
\begin{cases} 
&\mathbf{m}\xrightarrow{{\Phi}_{x,M}} \mathbf{m}'=\sqrt{x}\; {\bf m}\;,\\ \\
&V\xrightarrow{{\Phi}_{x,M}}V'=x  V + (2M+|1-x|)I_2\; ,
\end{cases} 
\end{align} with $\mathbf{m}'$ and $V'$ being respectively the first and second momenta of the output state ${\Phi}_{x,M}(\hat{\rho})$, and 
 with $I_2$ the $2\times 2$  identity matrix. 
For $x=\eta\in [0,1]$, and $M= (1-\eta)N$ with $N\geq 0$, the mapping~(\ref{defgenerale})
 corresponds to   the thermal attenuator channel $\mathcal{E}_{\eta,N}$
which 
describes the interaction of the a single-mode of the e.m. field  with an external
 thermal reservoir with $N$ mean photon number, mediated by a beam-splitter coupling of transmissivity $\eta$; for $x=g\geq 1$ and $M=(g-1) N$ with $N\geq 0$ instead, 
 ${\Phi}_{x,M}$ reduces to a thermal amplifier ${\cal A}_{g,N}$ which describes
the interaction 
with the input mode with  a thermal bath of mean photon number $N$, through a two mode squeezing operator with gain parameter~$g$; finally for $x=1$ and $M=N\geq 0$, $\Phi_{x,M}$  reduces to the 
 the additive classical noise GBC ${\cal N}_N$, i.e. 
\begin{equation} \label{attamp} 
\left\{ \begin{array}{llr} 
{\cal E}_{\eta,N} :={\Phi}_{x=\eta,M=(1-\eta)N},  & &\eta\in [0,1], N\geq 0 \;, \\ \\
{\cal N}_{N}:={\Phi}_{x=1,M=N}, &&  N\geq 0\;. \\\\
{\cal A}_{g,N}:={\Phi}_{x=g,M=(g-1)N}, && g\geq 1, N\geq 0  
 \;.\end{array} \right.
\end{equation} 

		\begin{center} 		
			\begin{table*}[t!]
\begin{tabular}{|c|ll|}
\hline
$(\mathbf{C_1})$ &$\mathcal{E}_{\eta_2,N_2}\circ \mathcal{E}_{\eta_1,N_1} =\mathcal{E}_{\eta_3,N_3}$ &
$\Big\{ \begin{array}{l} 
\eta_3= \eta_2\eta_1\\
(1-\eta_3) N_3 =  (1-\eta_2)N_2 + (1-\eta_1)\eta_2 N_1
\end{array} 
$
  \\    
	\hline  
	$(\mathbf{C_2})$ &  ${{\cal A}_{g_2,N_2}}\circ {{\cal A}_{g_1,N_1}} ={{\cal A}_{g_3,N_3}}$ &
$\Big\{  \begin{array}{l} 
g_3= g_2g_1\\
(g_3-1) N_3= (g_2-1)N_2 +(g_1-1)g_2 N_1
\end{array} $ 
\\
 \hline
$\begin{array}{c}(\mathbf{C_{3.1}})\\\\\\(\mathbf{C_{3.2}}) \end{array} 
$ &$\mathcal{E}_{\eta_2,N_{2}}\circ{\cal A}_{g_1,N_{1}} = 
\left\{ \begin{array}{l}
 \mathcal{E}_{\eta_3,N_3} \\ \\ \\
{\cal A}_{g_3,N_3} 
 \end{array} \right.$ &
$  \begin{array}{l}
\Big\{ \begin{array}{l}
\eta_3 =\eta_2 g_1  \\ 
{(1-\eta_3)(2N_3+1)=(1-\eta_2)(2N_{2}+1)+(\eta_3-\eta_2)(2N_{1}+1)}
\end{array} \\ \\
\Big\{ \begin{array}{l}  g_3 =\eta_2 g_1 \\  
{(g_3-1)(2N_3+1)=  (1-\eta_2)(2N_2+1)+ (g_3-\eta_2)(2N_1+1)} 
\end{array} 
\end{array} $\\
\hline
$\begin{array}{c}(\mathbf{C_{4.1}})\\\\\\(\mathbf{C_{4.2}}) \end{array} 
$ & ${\cal A}_{g_2,N_{2}} \circ \mathcal{E}_{\eta_1,N_{1}} = 
\left\{ \begin{array}{l}
 \mathcal{E}_{\eta_3,N_3} \\ \\ \\
{\cal A}_{g_3,N_3} 
 \end{array} \right.$ &
$\begin{array}{l}
\Big\{ \begin{array}{l} 
\eta_3 =g_2 \eta_1   \\ 
{(1-\eta_3)(2N_3+1)= (g_2-1)(2N_2+1) + (g_2-\eta_3)(2N_1+1)}
\end{array} \\ \\
\Big\{ \begin{array}{l}  g_3 =g_2 \eta_1  \\  
{(g_3-1)(2N_3+1)= (g_2-1)(2N_2+1) + (g_2-g_3)(2N_1+1)}\\
\end{array} 
\end{array} $\\
\hline
\end{tabular}
				\caption{Composition rules~(\ref{defNnew0}) and (\ref{defNnew}) expressed in terms of thermal attenuators and amplifiers via the identities~(\ref{attamp}).
				\label{tab1}} 
			\end{table*}
		\end{center} 	
 It is easy to check that the maps
${\Phi}_{x,M}$ are closed under concatenation, 
specifically given $(x_1,M_1), (x_2,M_2)\in  \mathbb{R}_+^2$ we have
that 
\begin{eqnarray} \label{defNnew0} 
{{\Phi}_{x_3,M_3}} = {{\Phi}_{x_2,M_2}}\circ {{\Phi}_{x_1,M_1}}\;, 
\end{eqnarray} 
is also a channel of the model with noise parameters $(x_3,M_3)\in  \left(\mathbb{R}^{+} \right)^{2}$ fulfilling the identities
\begin{eqnarray}
\begin{cases}
&x_3= x_2x_1\;, \\\\
&M_3 = M_2 + x_2 M_1 + \tfrac{|x_2-1| + x_2 |x_1-1| - |x_2x_1-1|}{2}\;,
 \label{defNnew} 
\end{cases} 
\end{eqnarray} 
which we express in terms of attenuators and amplifiers in Table~\ref{tab1}.

\section{A brief review on PI-GBC Capacities}\label{rev q cap}
\begin{figure*}[t!]
	\begin{tabular}{ c c }
		\includegraphics[width=0.99\columnwidth]{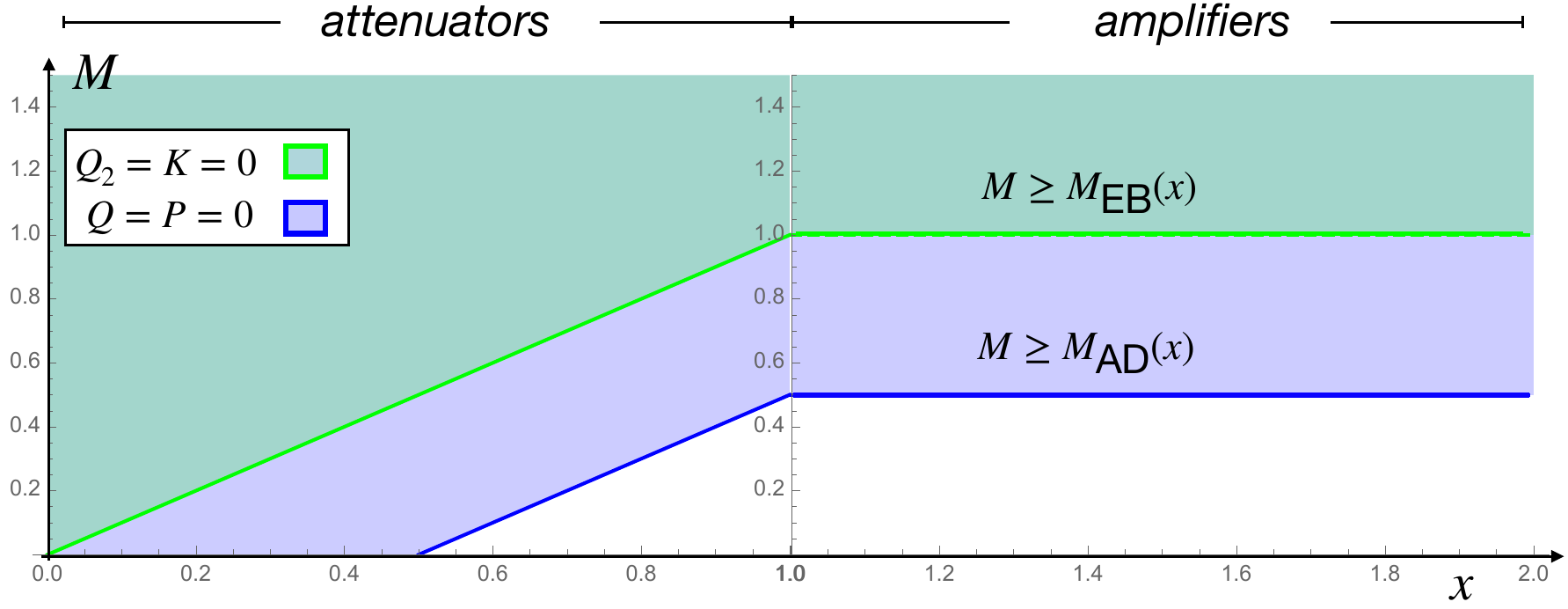} & \includegraphics[width=0.99\columnwidth]{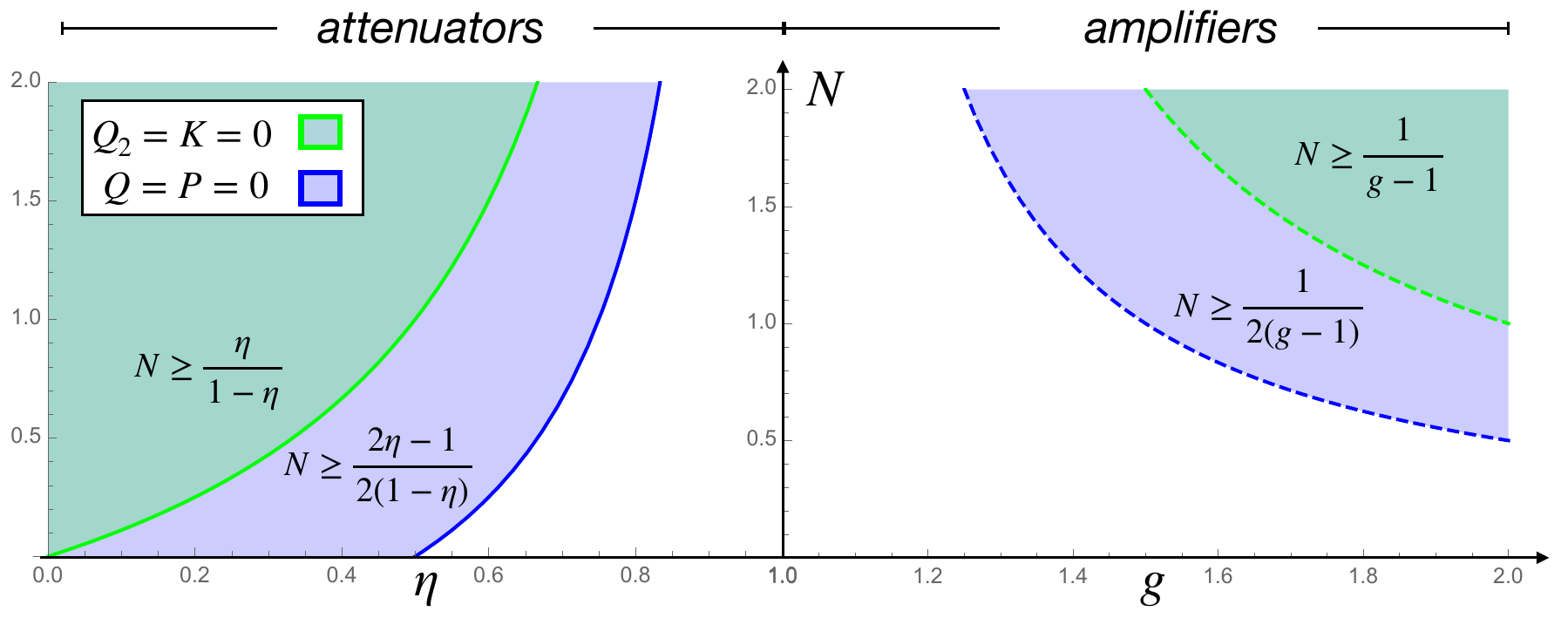}
	\end{tabular}
	\caption{Left panel: zero capacity regions for the PI-GBCs $\Phi_{x,M}$; Right panel: same plot expressed in terms of the parametrization (\ref{attamp}). 
	In both plots, the greenish areas represent the regions where the all the capacities ($Q_2$, $K$, $Q$, and $P$) are zero. The bluish areas represent the region
	where $Q$ and $P$ are zero but $Q_2$ and $K$ are not. Whether $Q$ and $P$ can  be zero  also for points in the white region is still an open problem. } 		\label{fig:zeros}
\end{figure*}

We start by recalling that 
for $N$ sufficiently large the channels 
$ {\cal E}_{\eta,N}$,  ${\cal N}_{N}$, and 
${\cal A}_{g,N}$ are Entanglement-Breaking (EB)~\cite{serafini,HOL BOOK}, specifically  
\begin{eqnarray} \left\{ 
\begin{array}{rcl}
 {\cal E}_{\eta,N}\equiv \text{EB}&  \Longleftrightarrow & \eta\in [0,1] \;, N \geq \frac{\eta}{1-\eta} \;, \\
{\cal N}_{N}\equiv \text{EB}&  \Longleftrightarrow & N \geq 1 \;, \\ 
{\cal A}_{g,N}\equiv \text{EB}&  \Longleftrightarrow &   g\geq 1\;, N \geq \frac{1}{g-1}\;.
 \end{array} \right.
\end{eqnarray}
In the  notation~(\ref{defgenerale}) this translates into the condition 
\begin{eqnarray} \label{EBregion1} 
\Phi_{x,M} \equiv \text{EB} \qquad
\Longleftrightarrow \qquad (x,M) \in \mathbb{EB} \;,
\end{eqnarray}
with the set 
\begin{eqnarray} \label{EBDD}
 \mathbb{EB}:= \{ (x,M)\in  \left(\mathbb{R}^{+} \right)^{2}: 
 M\geq M_{\text{EB}}(x)\}  \;,
\end{eqnarray}
defined by the threshold function 
    \begin{eqnarray}
 M_{\text{EB}}(x):=  \min\{ 1, x\} \;, \label{EBline}  
  \end{eqnarray}
(see Fig.~\ref{fig:zeros}).
By construction EB maps have all zero capacity values, i.e. 
\begin{eqnarray}
 (x,M) \in \mathbb{EB} \quad  
 \Longrightarrow \quad 
{\cal K}(\Phi_{x,M}) =0 \;, \label{zerocond} 
\end{eqnarray} 
with the implication that can be reversed for the two-way capacity  and for the secret-key capacity~\cite{MELE}, meaning  that $\mathbb{EB}$ 
corresponds to  the
 largest parameter region for which $Q_2$ and $K$ are null. The case of $Q$ and $P$ is different as it is known
 that these capacities nullify also for channels which do not belong to $\mathbb{EB}$.
 In particular one has  
 that
 \begin{equation} 
		\begin{cases} Q( \mathcal{E}_{\eta,N}) =P( \mathcal{E}_{\eta,N})  =0\;,  
		\quad \eta\in [0,1] \;, N \geq \tfrac{2\eta-1}{2(1-\eta)} 
		 \;, \label{thefirst} \\\\
		Q({\cal A}_{g,N}) =P({\cal A}_{g,N})  =0 \;,  
		\quad  g\geq 1 \;, N \geq \tfrac{1}{2(g-1)} \;, 
		\end{cases} 
			\end{equation}
			or 
	\begin{eqnarray} 
(x,M) \in \mathbb{AD} \quad  
 \Longrightarrow \quad 
{Q}(\Phi_{x,M}) = P(\Phi_{x,M})=0 \;, \label{zerocond00} \end{eqnarray}
with \begin{eqnarray} \label{EBDDD}
 \mathbb{AD}&:=& \{ (x,M)\in  \left(\mathbb{R}^{+} \right)^{2}: 
  M\geq M_{\text{AD}}(x)\}  \;, \end{eqnarray}
  \begin{eqnarray} \label{anti-deg conditions}
  M_{\text{AD}}(x)&:=&  \min\{ x-1/2,1/2\}\;, 
\end{eqnarray}		
corresponding to the Anti-Degradability (AD)
 region for the single-model PI-GBCs~\cite{CARUSO061,CARUSO06,HOLEVO07,LAMI19}
 (see Fig.~\ref{fig:zeros}). Notice 
 that at present it is still not clear whether or not ${\mathbb{AD}}$ is the largest set 
			where $Q$ and/or $P$ nullify.   
			What it is known are the exact values of these quantities  for   the special cases where $M=0$. Specifically in the case of the quantum and private capacities we have~\cite{WOLF}  			\begin{eqnarray}\label{wolf} 
			P(\Phi_{x,0}) = Q(\Phi_{x,0}) =\max\{0,\log_2\tfrac{x}{|1-x|}\}\;,
			\end{eqnarray} 
			while for the secret-key and two-way capacities it holds~\cite{Pirandola upp bound} 
			\begin{equation}\label{plob1} 
			K(\Phi_{x,0}) =Q_2(\Phi_{x,0})  = 
			\log_2(\tfrac{\max\{1,x\}}{|1-x|}) \;,
				\end{equation} 
				(notice that for amplifiers, i.e. $x\geq 1$, Eq.~(\ref{wolf}) and (\ref{plob1}) coincide). 

\begin{figure*}
     {{\includegraphics[width=0.99\columnwidth]{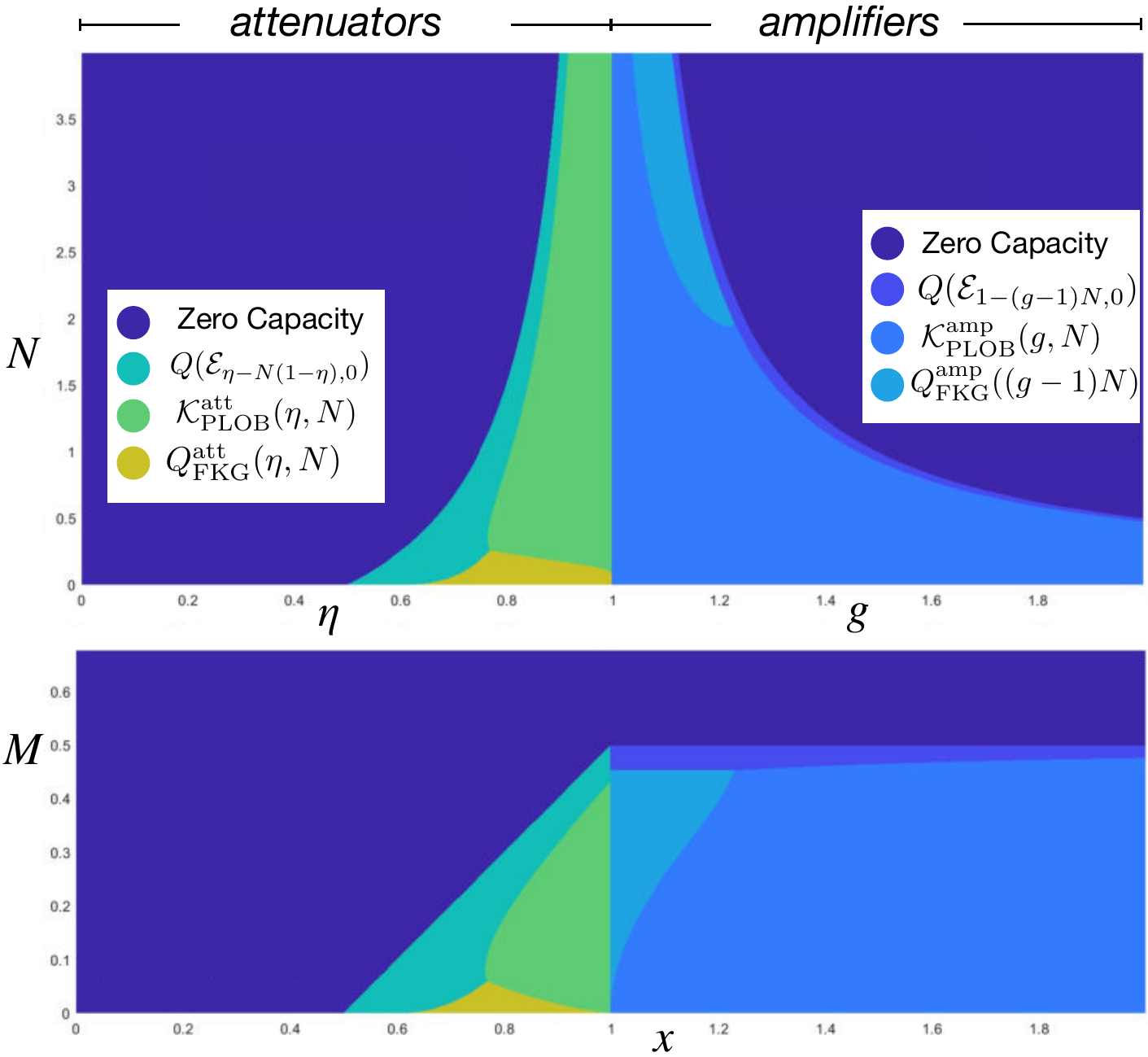} }}
     {{\includegraphics[width=0.99\columnwidth]{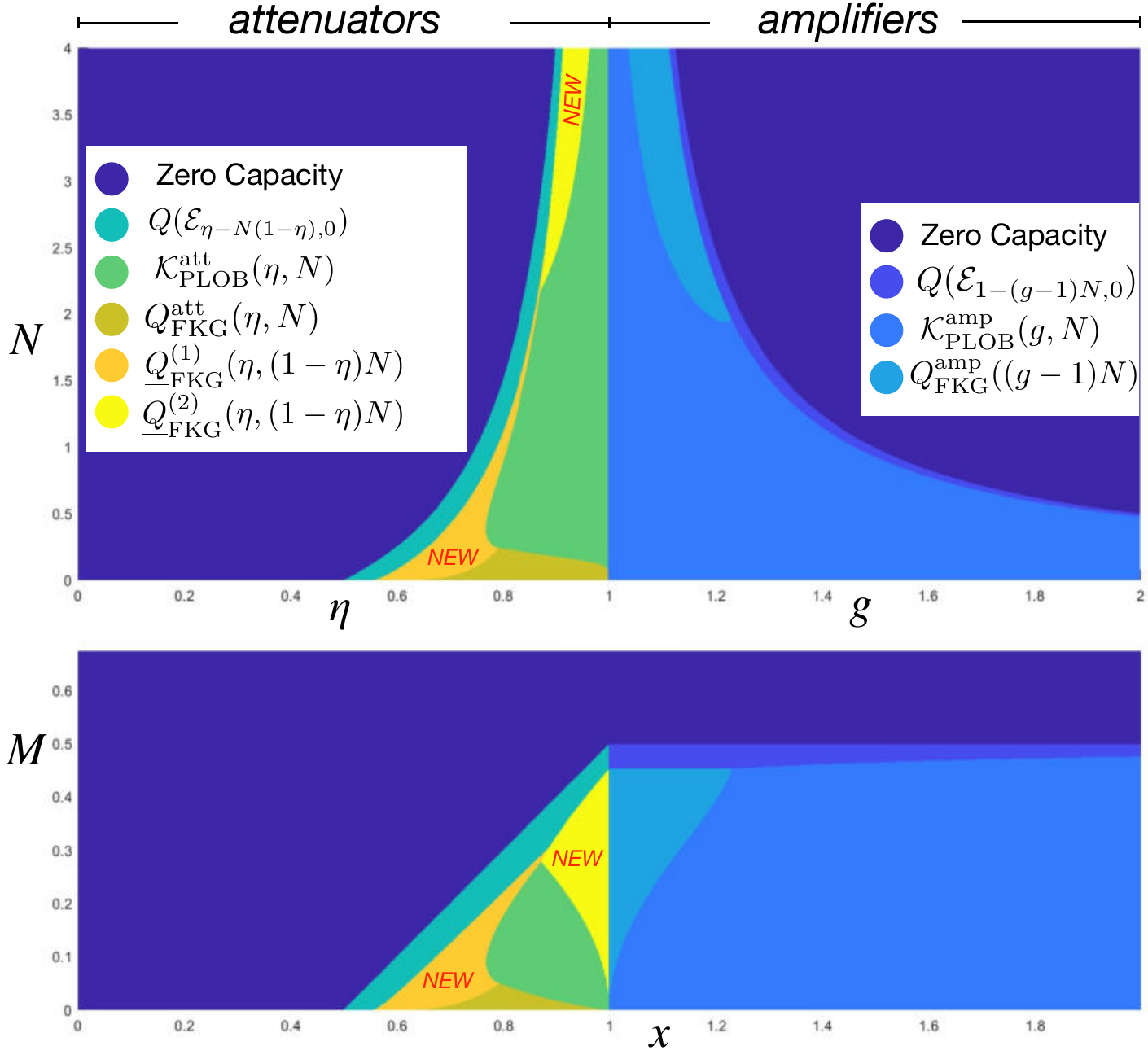} }}
    \caption{{\bf Left:} The top panel shows a numerical comparison between  the upper bounds (\ref{UPPLOB1}), (\ref{twist}), and (\ref{FKGup}) 
 for the quantum capacity  $Q$ of the channels ${\cal E}_{\eta, N}$ and ${\cal A}_{g,N}$. Each region with different colour indicates which 
    result is the best upper bound.  Purple region is where the quantum capacity is zero according to Eq.~(\ref{anti-deg conditions}).
  The bottom panel presents the same comparison expressed 
in terms of the $x,M$ parametrization~(\ref{defgenerale}). {\bf Right:} updated version of previous figure which includes the improved bounds  $\underline{Q}^{(1)}_\mathrm{FKG}(x,M)$ and $\underline{Q}^{(2)}_\mathrm{FKG}(x,M)$ of Eq.~(\ref{bounds with opt}): the orange and yellow regions (marked with the script {\it NEW})  are where they
    provide better constraints than the upper bounds of Sec.~\ref{sec:up} (notice that no improvement is obtained for amplifiers). The top panel reports the result in terms of the $\eta,N$ and $g,N$ parametrization, while the bottom panel reports the same result in the $x,M$ parametrization.} 
\label{fig:regions}\end{figure*}

\subsection{Upper bounds} \label{sec:up} 
State-of-the-art upper bounds for the capacities of thermal attenuators and amplifiers are given in Refs.~\cite{Pirandola upp bound,naj,matteo,flagged channel 3}. 
Specifically in \cite{Pirandola upp bound} it has been showed that, outside the EB region~(\ref{EBregion1}),
  all the quantum capacities ${\cal K}$ of thermal attenuators and
 amplifiers can be bounded as follows
\begin{equation}
\begin{cases} \label{UPPLOB1}
\begin{aligned}
{\cal K}(\mathcal{E}_{\eta,N})\leq {\cal K}^{\mathrm{att}}_{\mathrm{PLOB}}(\eta,N)&:= 	-h(N)-\log_2({\scriptstyle{(1-\eta)\eta^{N}}})\;, \\
\end{aligned}
\\ \\
\begin{aligned}
{\cal K}(\mathcal{A}_{g,N})\leq {\cal K}^{\mathrm{amp}}_{\mathrm{PLOB}}(g,N)&:=
	-h(N)+ \log_2(\tfrac{g^{N+1}}{g-1}) \;,\\ 
\end{aligned}
\end{cases}
\end{equation}
with 
\begin{eqnarray} h(x):=({x+1})\log_2({x+1})-x\log_2(x)\;,\end{eqnarray} 
(see also Ref.~\cite{btw} for a strong-converse extension of these inequalities).
Partial improvements w.r.t. to the above inequalities for the case of quantum and private capacities, have been reported in Refs.~\cite{matteo,naj} where special instances of the decomposition rules~$\mathbf{(C_{3.1})}$ and $\mathbf{(C_{3.2})}$ were employed to observe that outside the AD region~(\ref{thefirst}) (i.e. for $N \leq \tfrac{2\eta-1}{2(1-\eta)}$ for ${\cal E}_{\eta, N}$ and $N\leq \frac{1}{2 (g-1)}$ for
${\cal A}_{g,N}$)
one can write $\mathcal{E}_{\eta,N}=\mathcal{E}_{\eta-N(1-\eta),0}\circ{\cal A}_{\frac{\eta}{\eta-N(1-\eta)},0}$ 
and ${\cal A}_{g,N}=  \mathcal{E}_{1-(g-1)N ,0}\circ{\cal A}_{\frac{g}{1-(g-1)N },0}$,
which ultimately leads to 
\begin{equation}
\begin{cases}\label{twist}
\begin{aligned}
Q(\mathcal{E}_{\eta,N}),P(\mathcal{E}_{\eta,N})\leq Q(\mathcal{E}_{\eta-N(1-\eta),0})&=\log_2(\tfrac{\eta-N(1-\eta)}{1-\eta + N(1-\eta)}) \\
\end{aligned}
\\ \\ 
\begin{aligned}
Q({\cal A}_{g,N}),P(\mathcal{A}_{g,N})\leq  Q(\mathcal{E}_{1-(g-1)N,0})&=\log_2(\tfrac{1-(g-1)N}{(g-1)N})\; \\
\end{aligned}
\end{cases}
\end{equation}
(notice that the bounds nullifies at the border with the AD region).
In Ref.~\cite{flagged channel 3} instead, using degradable extensions of thermal attenuators, it was proven that 
\begin{eqnarray}
\label{FKGup} \begin{cases} 
Q(\mathcal{E}_{\eta,N}) ,P(\mathcal{E}_{\eta,N})\leq Q^{\mathrm{att}}_\mathrm{FKG}(\eta,N)\;,\\ \\
Q(\mathcal{A}_{g,N}), P(\mathcal{A}_{g,N})\leq Q^{\mathrm{amp}}_\mathrm{FKG}((g-1)N)\;,
  \end{cases} 
\end{eqnarray}
with 
\begin{eqnarray} \label{DEFFKBatt} 
Q^{\mathrm{att}}_\mathrm{FKG}(\eta,N) &:=&
  \log_2(\tfrac{\eta}{1-\eta})+h((1-\eta)N) -h(\eta N)\;,\nonumber\\ \nonumber
 Q^{\mathrm{amp}}_{\rm FKG}(M)&:=&-\log_2(e M)+2h \left(\tfrac{\sqrt{ M^2+1}-1}{2}\right)\;.
\end{eqnarray}
Notice that the first is a function which, for fixed $\eta$ is monotonically decreasing w.r.t. $N$ and, for fixed $N$, monotonically increasing in $\eta$ (being
null for $\eta \leq 0.5$). On the contrary, for $M=(g-1)N\in[0,1/2]$ (i.e. the only region where, in view of Eq.~(\ref{thefirst}) 
 it makes sense to 
consider~(\ref{FKGup})) $Q^{\mathrm{amp}}_{\rm FKG}(M)$  is a positive, monotonically increasing function.
As explicitly shown in the left part of  Fig.~\ref{fig:regions}, one may notice 
 that while for low values of $\eta$ the upper bound~(\ref{twist}) outperform the others, as  $\eta$ approaches $1$, (\ref{UPPLOB1}) and (\ref{FKGup}) 
 win (in particular $Q^{\mathrm{att}}_\mathrm{FKG}$ provides
 the best bound in the low noise regime $N\ll 1$, while $Q^{\mathrm{att}}_\mathrm{PLOB}$ does it for higher $N$). 
 
 \subsection{Lower bounds} 
 Lower bounds for the two-way and secret capacities are provided by the inequalities~\cite{PIR09,HOLEVO01}  
 \begin{equation} \label{lower1}
\begin{cases}
K(\mathcal{E}_{\eta,N})	\geq {Q_2}(\mathcal{E}_{\eta,N})\geq\max\{ 0, 	-h(N)-\log_2(1-\eta)\}\;, \\\\
	K(\mathcal{A}_{g,N})	\geq {Q_2}(\mathcal{A}_{g,N})\geq \max\{ 0,
	-h(N)+ \log_2(\tfrac{g}{g-1})\}\, , 
	\end{cases} 
\end{equation}
{
 which have been improved in~\cite{PIR16,MELE}. \cite{PIR16} improved the lower bound for secret capacity adopting Gaussian protocols based on suitable trusted-noise detectors, and \cite{MELE} showed that the region where $Q_2$ and $K$ are non-zero
 extend beyond what predicted by the above equations, including all the non-EB region.} 
A lower bound on $Q$ and  $P$ is given instead by the coherent information for one use of the channel, evaluated on an
infinite temperature state
\begin{equation}\label{lower2} 
	Q(\mathcal{E}_{\eta,N})\geq Q^{\mathrm{att}}_\mathrm{low}(\eta,N)=\mathrm{max}\{0,\log_2(\tfrac{\eta}{1-\eta})-h(N)\}\, .
\end{equation}

\section{Low ground/high ground capacity regions analysis}\label{l and h reg}
Using the composition rules~(\ref{defNnew0}), (\ref{defNnew})  together with the data-processing inequality~(\ref{data}),
  in the parameter space of the maps~(\ref{defgenerale})  we can identify regions  where the capacities are
	provably smaller or larger than an assigned reference value. 
	For this purpose given $(x,M)\in \left(\mathbb{R}^{+} \right)^{2}$  we define ${\mathbb{L}}_{x,M}:={\mathbb{L}}(\Phi_{x,M})$ the collection of points  $(x',M')\in \left(\mathbb{R}^{+} \right)^{2}$ such that
	such that the channel
 $\Phi_{x',M'}$ can be decomposed as a three-elements concatenation $ \Phi_{x',M'} =\Phi_{\bar{x}_1,\bar{M}_1} \circ  \Phi_{x,M} \circ \Phi_{\bar{x}_2,\bar{M}_2}$ that involves   
$\Phi_{x,M}$ together with two other maps $\Phi_{\bar{x}_1,\bar{M}_1}$ and $\Phi_{\bar{x}_2,\bar{M}_2}$, i.e.~\cite{NOTA1} 
\begin{eqnarray}
&&\!\!\!\!\!{\mathbb{L}}_{x,M} := \Big\{ (x',M')  \in \left(\mathbb{R}^{+} \right)^{2} :\nonumber 
\exists (\bar{x}_1,\bar{M}_1),  (\bar{x}_2,\bar{M}_2) \in \left(\mathbb{R}^{+} \right)^{2}  
 \\  \label{alternativedefLO} 
&& \qquad \qquad \Phi_{x',M'} =\Phi_{\bar{x}_1,\bar{M}_1} \circ  \Phi_{x,M} \circ \Phi_{\bar{x}_2,\bar{M}_2} 
 \Big\} \;. \end{eqnarray}
From~(\ref{data}) 
 it turns out that 
\begin{eqnarray} \label{prop1} 
{\cal K}(\Phi_{x',M'}) \leq {\cal K}(\Phi_{x,M})\;, \qquad \forall (x',M') \in {\mathbb{L}}_{x,M}\;.
\end{eqnarray} 
so we dub ${\mathbb{L}}_{x,M}$ the {\it low-ground 
	capacity region} of the channel $\Phi_{x,M}$.  Notice in particular that, since 
	$\Phi_{0,M'} \circ \Phi_{x,M} = \Phi_{0,M'}$ 
	 for all $(x,M)$ and $M'\geq 0$,  one has that all the points
	$(0,M')$ are included into   ${\mathbb{L}}_{x,M}$, i.e.
	\begin{eqnarray} \label{trivial} 
	(0,M') \in {\mathbb{L}}_{x,M} \;, \quad \forall M'\geq 0\;, \forall (x,M) \in  \left(\mathbb{R}^{+} \right)^{2}\;.
	\end{eqnarray} 
By reversing the ordering of the concatenations in Eq.~(\ref{alternativedefLO}) 
 we also introduce the {\it high-ground 
	capacity region} ${\mathbb{H}}_{x,M}:={\mathbb{H}}(\Phi_{x,M})$  of the channel $\Phi_{x,M}$, i.e. 
	\begin{eqnarray}
&&\!\!\!\!\!{\mathbb{H}}_{x,M} := \Big\{ (x',M')  \in \left(\mathbb{R}^{+} \right)^{2} :\nonumber 
\exists (\bar{x}_1,\bar{M}_1),  (\bar{x}_2,\bar{M}_2) \in \left(\mathbb{R}^{+} \right)^{2}  
 \\  \label{alternativedefHIGH} 
&& \qquad \qquad \Phi_{x,M} =\Phi_{\bar{x}_1,\bar{M}_1} \circ  \Phi_{x',M'} \circ \Phi_{\bar{x}_2,\bar{M}_2} 
 \Big\} \;, \end{eqnarray}
 which  fulfils the condition 
 \begin{eqnarray} \label{prop2} 
{\cal K}(\Phi_{x',M'}) \geq {\cal K}(\Phi_{x,M})\;, \qquad \forall (x',M') \in {\mathbb{H}}_{x,M}\;.
\end{eqnarray} 
Notice that by construction ${\mathbb{L}}_{x,M}$ and ${\mathbb{H}}_{x,M}$ obey a natural ordering 
\begin{eqnarray}\label{Hide11} 
 {\mathbb{L}}_{x',M'} &\subseteq&  {\mathbb{L}}_{x,M} \qquad \forall (x',M') \in {\mathbb{L}}_{x,M} \;, \\  \label{Hide22} 
 {\mathbb{H}}_{x',M'} &\subseteq&  {\mathbb{H}}_{x,M} \qquad \forall (x',M') \in {\mathbb{H}}_{x,M} \;,
\end{eqnarray}  
and fulfil the complementary relation
\begin{eqnarray} \label{complementary} 
(x',M') \in {\mathbb{H}}_{x,M} \quad \Longleftrightarrow  \quad (x,M) \in {\mathbb{L}}_{x',M'} \;.
\end{eqnarray}

  In the next subsections we shall 
  provide an  analytic characterization of  ${\mathbb{L}}_{x,M}$ and 
  ${\mathbb{H}}_{x,M}$. As we shall see one can identify two different regimes ruled by the function $M_{\text{EB}}(x)$ of 
  Eq.~(\ref{EBline}) which  identifies the EB sector.
Indeed for points 
  $(x,M)$ with 
  \begin{eqnarray}
  M \leq  M_{\text{EB}}(x)=  \min\{ 1, x\} \;, \label{notEB}  
  \end{eqnarray} 
  that is for channels which are non-EB and for the EB ones which are on the border line of the region $\mathbb{EB}$, 
the sets ${\mathbb{L}}_{x,M}$ and ${\mathbb{H}}_{x,M}$
 are defined by the 
 functions 
  \begin{eqnarray}
  f^{(1)}_{x,M}(x') &: =& M  + (1-x) \Theta(1-x)
+ (x'-1) \Theta(1-x')  \;,\nonumber  \\ 
   f^{(2)}_{x,M}(x') &: =& (x'/x) \big[ M
  +(x-1) \Theta(x-1)\big]    \nonumber \\ &&  \label{effe2}
  \qquad\qquad - (x'-1)\Theta(x'-1)  \;,
  \end{eqnarray} 
  with $\Theta(x)$ being the Heaviside step-function.
  Specifically we shall prove that under the condition~(\ref{notEB}),
   ${\mathbb{L}}_{x,M}$ is formed by all points  which are above $ f^{(1)}_{x,M}(x')$ and $ f^{(2)}_{x,M}(x')$ , i.e. 
    \begin{eqnarray} \label{analytical1} 
  {\mathbb{L}}_{x,M} &=&  
  \Big\{ (x',M')\in\left(\mathbb{R}^{+} \right)^{2} :  
  \\ \nonumber && \quad  M'\geq \max\{  f^{(1)}_{x,M}(x'),  f^{(2)}_{x,M}(x')\} \Big\} \;,
  \end{eqnarray} 
  while ${\mathbb{H}}_{x,M}$ is given by the polytope formed   by the points below such curves, i.e. 
     \begin{eqnarray} \label{analytical2} 
  {\mathbb{H}}_{x,M} &=&  
  \Big\{ (x',M')\in\left(\mathbb{R}^{+} \right)^{2} :  
  \\ \nonumber && \quad  M'\leq \min\{  f^{(1)}_{x,M}(x'),  f^{(2)}_{x,M}(x')\} \Big\} \;.
  \end{eqnarray}
As evident from Fig.~\ref{figurnew3bis},  the domains identified by Eqs.~(\ref{analytical1}) and (\ref{analytical2}) 
 admit  $(x,M)$ as unique contact point, implying that for a relative large portion  of the phase space 
 $\left(\mathbb{R}^{+} \right)^{2}$ we cannot assign a definite ordering w.r.t.  ${\cal K}(\Phi_{x,M})$ (white regions of plots). 
 The situation change however when  the map $\Phi_{x,M}$ is deep inside  the EB region, i.e. for    \begin{eqnarray}
  M >  M_{\text{EB}}(x)=  \min\{ 1, x\} \;. \label{EBconstraint}  
  \end{eqnarray} 
Under the constraint~(\ref{EBconstraint})  the sets ${\mathbb{H}}_{x,M}$  and ${\mathbb{L}}_{x,M}$
 provide a complete covering of the phase space and  have a non trivial overlap
 ${\mathbb{O}}_{x,M} := {\mathbb{H}}_{x,M} \bigcap {\mathbb{L}}_{x,M}$.  Indeed,
  one can show that
  irrespectively from the specific choice of $(x,M)$, 
  ${\mathbb{H}}_{x,M}$ coincides with the entire space $\left(\mathbb{R}^{+} \right)^{2}$, i.e.
   \begin{equation} \label{analytical2new11} 
{\mathbb{H}}_{x,M} =  \left(\mathbb{R}^{+} \right)^{2}\;,
  \end{equation}
  while  ${\mathbb{L}}_{x,M}$ (and hence ${\mathbb{O}}_{x,M}$) corresponds to the $x\rightarrow 0, M\rightarrow 0$,
  limit of Eq.~(\ref{analytical1}), i.e.~\cite{nota0}       \begin{eqnarray} \label{analytical2new} 
  {\mathbb{L}}_{x,M} &=&   {\mathbb{O}}_{x,M}  \\\nonumber
  &=& {\mathbb{L}}_{0,0} :=
  \Big\{ (x',M')\in\left(\mathbb{R}^{+} \right)^{2} :  M' > M_{\text{EB}}(x')\Big\} \;,
  \end{eqnarray}
which coincides with the subset identified by Eq.~(\ref{EBconstraint}). 
This implies  that given any two points $(x_1,M_1)$ and $(x_2,M_2)$ in   
 ${\mathbb{L}}_{0,0}$, 
their associated channel  are 
 equivalent up to concatenation with extra GBCs~(\ref{defgenerale}), i.e.  there exist 
proper choices of the maps $\Phi_{\bar{x}_1,\bar{M}_1}$, $\Phi_{\bar{x}_2,\bar{M}_2}$, $\Phi_{\bar{x}_3,\bar{M}_3}$,  and $\Phi_{\bar{x}_4,\bar{M}_4}$, such that we can write 
\begin{equation}\begin{cases}
 \Phi_{x_1,M_1} =\Phi_{\bar{x}_1,\bar{M}_1} \circ  \Phi_{x_2,M_2} \circ \Phi_{\bar{x}_2,\bar{M}_2} \;,\\
 \Phi_{x_2,M_2} =\Phi_{\bar{x}_3,\bar{M}_3} \circ  \Phi_{x_1,M_1} \circ \Phi_{\bar{x}_4,\bar{M}_4} \;,
 \end{cases} \label{equivalent} 
 \end{equation} 
 which in turn imposes  ${\cal K}(\Phi_{x_1,M_1}) ={\cal K}(\Phi_{x_2,M_2})$ in agreement with the property~(\ref{zerocond}). 

\begin{figure}[t!]
		\includegraphics[width=1\columnwidth]{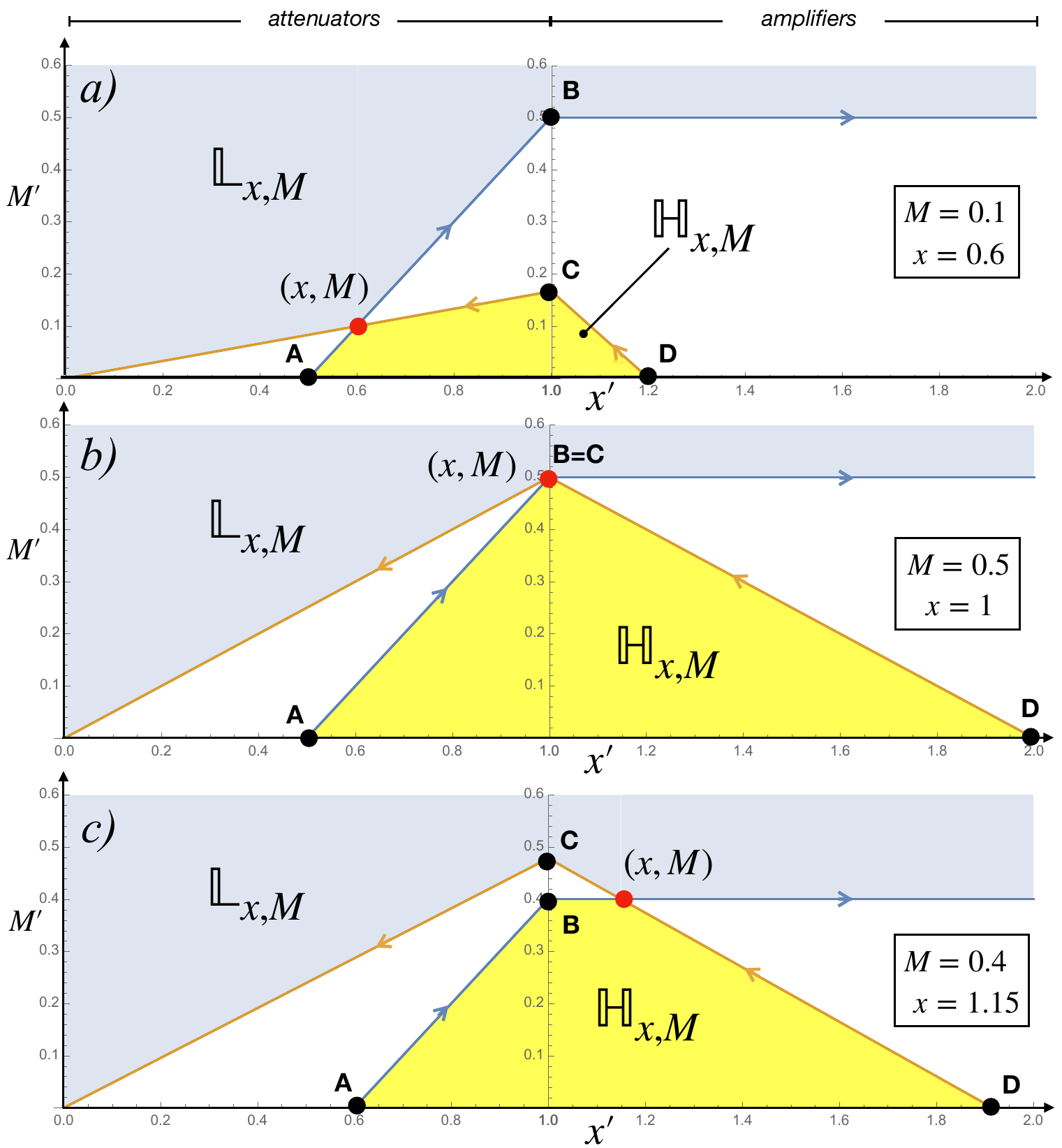}\\
		\includegraphics[width=1\columnwidth]{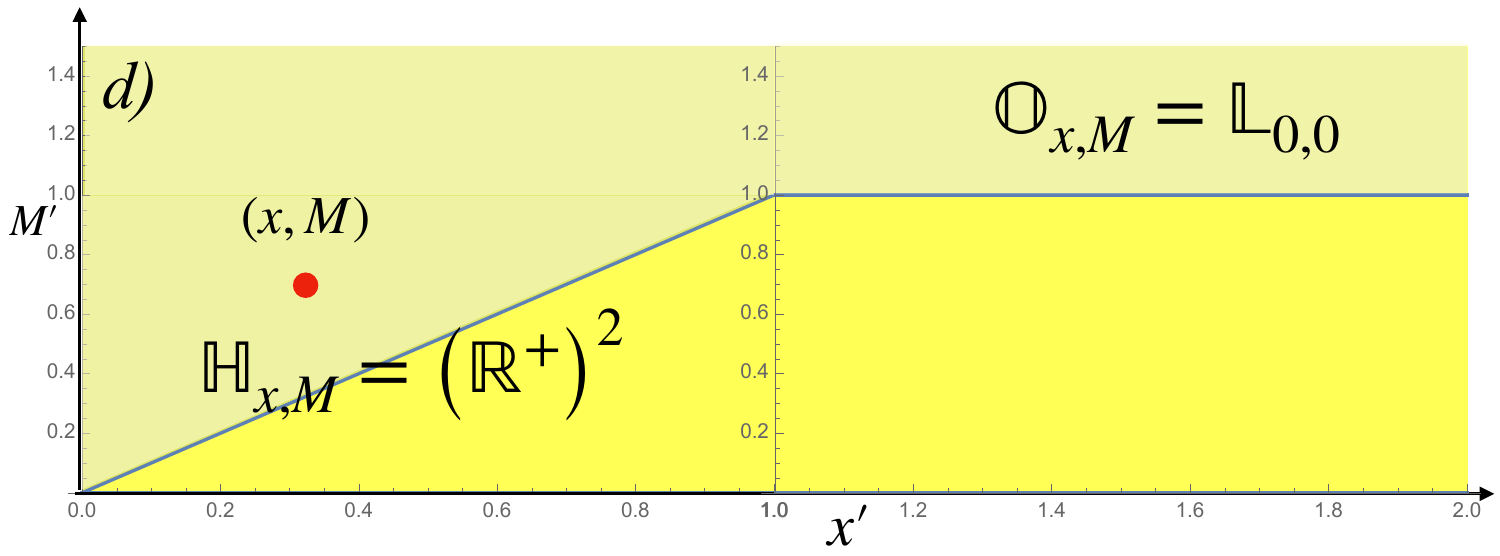}
	\caption{Low-ground  capacity region ${\mathbb{L}}_{x,M}$~(\ref{alternativedefLO})  (light blue area)  and high-ground capacity region
	${\mathbb{H}}_{x,M}$~(\ref{alternativedefHIGH}) (yellow)  for the channel $\Phi_{x,M}$ (red dot element).
 The first three top panels   refer to the regime~(\ref{notEB})
 where ${\mathbb{L}}_{x,M}$ and ${\mathbb{H}}_{x,M}$ are expressed respectively by Eqs.~(\ref{analytical1}) and~(\ref{analytical2}): specifically in panel a)  the reference map is an 
attenuator ($x=0.6$, $M=0.1$), in panel b) it is an additive classical noise
map ($x=1$, $M=0.5$), and in panel c)  it is an amplifier ($x=1.15$, $M=0.1$). Panel d) instead presents  a case where $\Phi_{x,M}$ fulfils the EB condition~(\ref{EBconstraint}); here ${\mathbb{H}}_{x,M}$ coincides with the full plane   $\left(\mathbb{R}^{+} \right)^{2}$, while ${\mathbb{L}}_{x,M}$ (and hence the overlap ${\mathbb{O}}_{x,M}={\mathbb{H}}_{x,M}\bigcap {\mathbb{L}}_{x,M}$),
are given by the set ${\mathbb{L}}_{0,0}$ of Eq.~(\ref{analytical2new}): by construction any two points
in this area are equivalent under GBC concatenation, see Eq.~(\ref{equivalent}). 
	The border lines which identify the various  regions are the functions $f^{(1)}_{x,M}(x')$ (blue curve) and $f^{(2)}_{x,M}(x')$ (orange curve) defined in Eq.~(\ref{effe2}) -- for panel d) 
	the border is given by $f^{(1)}_{0,0}(x')=\min\{1,x'\}$~\cite{nota0}.
The arrows in the plots show the direction where
	 the capacities have to decrease (or at most remain constant) and the white regions describe portions of the
	 parameter space where  the composition rules cannot be used to determine a specific capacity ordering with respect to the reference channel. 
	 In panel a) we have ${\bf A}=(x-M,0)$, ${\bf B}=(1,M-x+1)$, ${\bf C}=(1,M/x)$, ${\bf D}=(x/(x-M),0)$; in panel b) and c) instead 
	 ${\bf A}=(1-M,0)$, ${\bf B}=(1,M)$, ${\bf C}=(1,(M+x-1)/x)$, ${\bf D}=(x/(1-M),0)$ (notice that for b), ${\bf B}={\bf C}=(x,M)$). 
	 } 		\label{figurnew3bis}
\end{figure}
\section{Analytical characterization of ${\mathbb{L}}_{x,M}$ and ${\mathbb{H}}_{x,M}$ } 
In this section we give an analytical characterization of the low-ground and high-ground regions 
for an arbitrary channel $\Phi_{x,M}$. 
We start in Sec.~\ref{two-times} by focusing on a simplified version of the 
 concatenation rules entering in the definitions~(\ref{alternativedefLO}) and (\ref{alternativedefHIGH}) where  $\Phi_{x,M}$  is connected with the elements 
 $\Phi_{x',M'}$ via only a single extra PI-GBC element $\Phi_{\bar{x},\bar{M}}$. This will allows us to identify  two regions  
  \begin{eqnarray}\label{alternativedefLOsimple}
&&\!\!\!\!\!{\mathbb{L}}^{(0)}_{x,M} := \Big\{ (x',M')  \in \left(\mathbb{R}^{+} \right)^{2} :
\exists (\bar{x},\bar{M})\in \left(\mathbb{R}^{+} \right)^{2}  
 \\  
&& \Phi_{x',M'} =\Phi_{\bar{x},\bar{M}} \circ  \Phi_{x,M} \; \mbox{or} 
\; \  \Phi_{x',M'} =  \Phi_{x,M} \circ \Phi_{\bar{x},\bar{M}} 
 \Big\} \;,  \nonumber  \end{eqnarray}
 and 
 	\begin{eqnarray}  \label{alternativedefHIGHsimple}
&&\!\!\!\!\!{\mathbb{H}}^{(0)}_{x,M} := \Big\{ (x',M')  \in \left(\mathbb{R}^{+} \right)^{2} :
\exists (\bar{x},\bar{M}) \in \left(\mathbb{R}^{+} \right)^{2}  
 \\ 
 && \Phi_{x,M} =\Phi_{\bar{x},\bar{M}} \circ  \Phi_{x',M'} \; \mbox{or} 
\; \  \Phi_{x,M} =  \Phi_{x',M'} \circ \Phi_{\bar{x},\bar{M}} 
 \Big\} \;. \nonumber \end{eqnarray}
 which by  construction are subsets of ${\mathbb{L}}_{x,M}$ and ${\mathbb{H}}_{x,M}$, i.e.
 \begin{eqnarray} \label{impoinclusions} 
{\mathbb{L}}^{(0)}_{x,M} \subseteq {\mathbb{L}}_{x,M}\;, \qquad   {\mathbb{H}}^{(0)}_{x,M} \subseteq   {\mathbb{H}}_{x,M}\;. \end{eqnarray}

In Sec.~\ref{three-times}  we shall prove that for $(x,M)$ fulfilling the constraint~(\ref{notEB}),  ${\mathbb{L}}^{(0)}_{x,M}$ and ${\mathbb{H}}^{(0)}_{x,M}$
 indeed coincide with  ${\mathbb{L}}_{x,M}$ and ${\mathbb{H}}_{x,M}$ leading to Eqs.~(\ref{analytical1}) and (\ref{analytical2}). 
 The derivation of Eqs.~(\ref{analytical2new11}) and (\ref{analytical2new}) for maps 
 not fulfilling~(\ref{notEB}) will instead be given in Sec.~\ref{three-timesEB}.

\subsection{Two-elements concatenations}~\label{two-times} 
To determine ${\mathbb{L}}^{(0)}_{x,M}$ and ${\mathbb{H}}^{(0)}_{x,M}$ 
 we  adopt the parametrization~(\ref{attamp}) 
to better underline the role played by amplifiers and attenuators. 
Given  hence $\eta\in [0,1]$ and $N\geq 0$, 
let us introduce the following functions 
\begin{eqnarray} \label{defnattnamp}
	{N}_{\eta,N}^{{(1)}}(\eta')&:=&N \left(\tfrac{1-\eta}{\eta}\right)\left(\tfrac{\eta' }{1-\eta'}\right)\;, 
	\\ {N}_{\eta,N}^{(2)}(\eta')&:=& (N+1)\left(\tfrac{1-\eta}{1-\eta'}\right)-1\;.
	\label{defnattnamp11}
	\end{eqnarray}
It then turns out that 
	\begin{theorem}\label{th3}
	The attenuator map  ${\cal E}_{\eta,N}$ admits
\begin{eqnarray}\label{lowgrounddef} 
\begin{cases}
 {\mathbb{L}}^{({\rm att},1)}_{\eta,N} &:=\left\{ 
  {\scriptstyle{(\eta',N') : 0 \leq \eta'\leq \eta \;,
N'\geq \max\{  {N}_{\eta,N}^{{(1)}}(\eta')}},0\} \right\}, \\
 {\mathbb{L}}^{({\rm att},2)}_{\eta,N}  &:= \left\{  {\scriptstyle{ (\eta',N') :   1\geq \eta'\geq \eta \;,
N'\geq  \max\{  {N}_{\eta,N}^{(2)}(\eta')}},0\} \right\},  \\ 
	 {\mathbb{L}}^{({\rm att})}_{\eta,N} &:= {\mathbb{L}}^{({\rm att},1)}_{\eta,N}  \cup  {\mathbb{L}}^{({\rm att},2)}_{\eta,N} \;, 
	\end{cases} 
	\end{eqnarray}  
	(light blue area on the left-hand-side of the top panel of Fig.~\ref{figurnew3}),
	as subset of the corresponding two-element concatenation, low-ground capacity region 
	${\mathbb{L}}^{(0)}({\cal E}_{\eta,N})$,
	and 
	\begin{eqnarray}
\begin{cases}
 {\mathbb{H}}^{({\rm att},1)}_{\eta,N} &:=\left\{  {\scriptstyle{ (\eta',N'): 1\geq \eta'\geq \eta \;,
0 \leq N'\leq  {N}_{\eta,N}^{{(1)}}(\eta')}}\right\} \;, \\
 {\mathbb{H}}^{({\rm att},2)}_{\eta,N}  &:= \left\{  {\scriptstyle{ (\eta',N'): 1\leq \eta'\leq \eta \;,
0\leq N'\leq   {N}_{\eta,N}^{(2)}(\eta')}}\right\} \;,  \\ 
	 {\mathbb{H}}^{({\rm att})}_{\eta,N}  &:= {\mathbb{H}}^{({\rm att},1)}_{\eta,N}  \cup {\mathbb{H}}^{({\rm att},2)}_{\eta,N} \;, 
	\end{cases} 
	\end{eqnarray}  
	(yellow area on the~left-hand-side of top panel of Fig.~\ref{figurnew3})
	as subset of  ${\mathbb{H}}^{(0)}({\cal E}_{\eta,N})$, i.e.
	\begin{eqnarray} {\mathbb{L}}^{({\rm att})}_{\eta,N}  \subseteq {\mathbb{L}}^{(0)}({\cal E}_{\eta,N})\;, \qquad 
	{\mathbb{H}}^{({\rm att})}_{\eta,N}  \subseteq {\mathbb{H}}^{(0)}({\cal E}_{\eta,N})\label{IMPO12}\;.\end{eqnarray}
\end{theorem}
\begin{proof}
To derive the first inclusion  of Eq.~(\ref{IMPO12})
 we set $(\eta_3,N_3)= (\eta',N')$ and $(\eta_1,N_1)= (\eta,N)$ in Eq.~$\mathbf{(C_1)}$ of Tab.~\ref{tab1}  to observe that 
 for all $\eta_2\in[0,1]$ and $N_2\geq 0$, 
 \begin{eqnarray} \label{dec1} \mathcal{E}_{\eta',N'}=\mathcal{E}_{\eta_2,N_2}\circ \mathcal{E}_{\eta,N}\;, 
 \end{eqnarray} 
is also a thermal channel with parameters 
\begin{eqnarray}
\begin{cases} \eta'= \eta\eta_2 \leq \eta \;, \\
N'=\tfrac{(1-\eta)\eta_2 N + (1-\eta_2)N_2}{1-\eta\eta_2}\geq  \tfrac{(1-\eta)\eta_2 }{1-\eta\eta_2}N = {N}_{\eta,N}^{{(1)}}(\eta')\;,
\end{cases} 
\label{defN'} 
\end{eqnarray} 
that span the entire set ${\mathbb{L}}^{({\rm att},1)}_{\eta,N}$, as the inequality is saturated whenever $N_2=0$. 
Therefore in view of the
definition~(\ref{alternativedefLO}) we can conclude that 
\begin{eqnarray} \label{primaprima} 
{\mathbb{L}}^{({\rm att},1)}_{\eta,N}  \subseteq {\mathbb{L}}^{(0)}({\cal E}_{\eta,N})\;. 
\label{IMPO1}
\end{eqnarray} 
We next invoke Eq.~$\mathbf{(C_{3.1})}$ of Tab.~\ref{tab1}  to observe that 
\begin{eqnarray} \label{dec2} 
{\cal E}_{\eta',N'}=\mathcal{E}_{\eta,N}\circ{\cal A}_{g_1,N_1}\;,\end{eqnarray} 
with 
\begin{eqnarray}
\begin{cases} 
\eta'= g_1\eta\geq \eta \;, \\\\
N'=\tfrac{(1-\eta)(2 N+1) + \eta(g_1-1)(2 N_1 +1)-(1-\eta')}{2(1-\eta')}\\
\qquad\quad \geq  \tfrac{\eta'-\eta + (1-\eta)N}{1-\eta'}
={N}_{\eta,N}^{(2)}(\eta')\label{defN'mew} \;,
\end{cases} 
\end{eqnarray} 
is a thermal map too which spans the entire subset ${\mathbb{L}}^{({\rm att},2)}_{\eta,N}$, as the inequality is saturated whenever $N_1=0$. Accordingly we can claim  that 
\begin{eqnarray} \label{IMPO2}  
{\mathbb{L}}^{({\rm att},2)}_{\eta,N}   \subseteq {\mathbb{L}}^{(0)}({\cal E}_{\eta,N})\;, 
\end{eqnarray} 
which together with (\ref{IMPO1}) yields the first identity of Eq.~(\ref{IMPO12}). 
\begin{figure}[t!]
		\includegraphics[width=1.00\columnwidth]{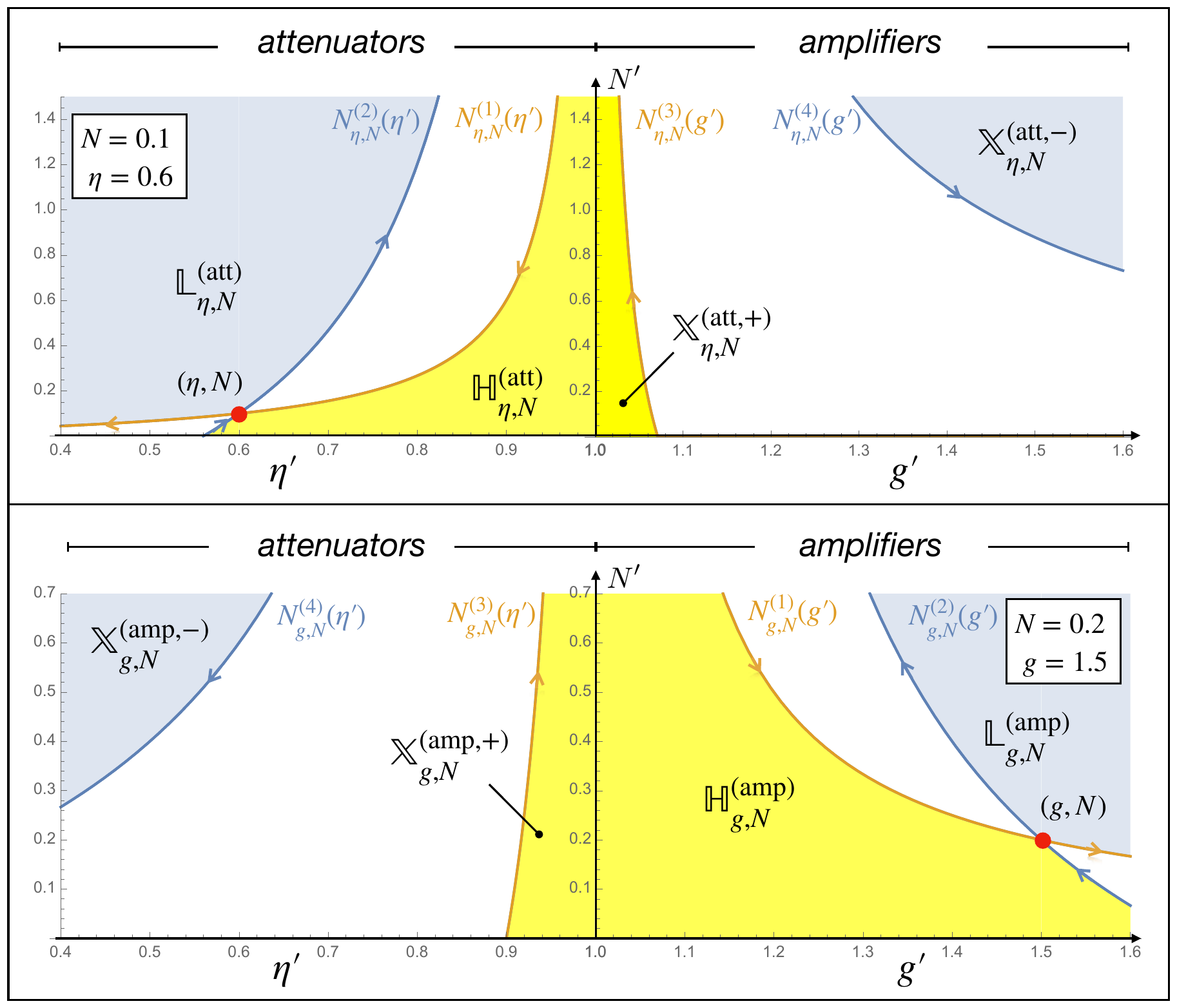} 
	\caption{Two-element concatenation analysis of the low-ground and high-ground capacity regions for thermal attenuators and 
	amplifiers (for all the examples reported in the plots the channels are non-EB, i.e. fulfil the condition~(\ref{notEB})). 
Top panel:  given a thermal attenuator ${\cal E}_{\eta, N}$ (red dot element), 
the light blue areas correspond respectively to 
${\mathbb{L}}^{({\rm att})}_{\eta,N}$ and  ${\mathbb{X}}^{({\rm att},-)}_{\eta,N}$; 
	the yellow areas describe instead  ${\mathbb{H}}^{({\rm att})}_{\eta,N}$  and ${\mathbb{X}}^{({\rm att},+)}_{\eta,N}$.
	 Plot realized using 
	$N=0.1$ and $\eta=0.6$. {Bottom panel}:  
	given a thermal amplifier ${\cal A}_{g, N}$ (red dot element), 
the light blue area   describe the sets ${\mathbb{L}}^{({\rm amp})}_{g,N}$ and ${\mathbb{X}}^{({\rm amp},-)}_{g,N}$ while 
	the yellow area  ${\mathbb{H}}^{({\rm amp})}_{g,N}$  and ${\mathbb{X}}^{({\rm amp},+)}_{g,N}$.
	Plot realized using 
	$N=0.2$ and $g=1.5$.
The orange and blue curves represent the border lines of the low ground/high ground regions defined by the  functions ${N}_{\eta,N}^{{(1,2)}}(\eta')$ of
	Eqs.~(\ref{defnattnamp}), (\ref{defnattnamp11}),  ${N}_{\eta,N}^{{(3,4)}}(g')$ of Eqs.~(\ref{defn3g}), (\ref{defn3gw}),
	${N}_{g,N}^{{(1,2)}}(g')$ of
	Eqs.~(\ref{Ndefnattnamp}), (\ref{Ndefnattnamp11}),  and ${N}_{g,N}^{{(3,4)}}(\eta')$ of Eqs.~(\ref{ddf}), (\ref{ddf1}):
	along these curves the  
	arrows show the direction where 
	 the capacities have to decrease (or at most remain constant) -- see Corollaries~\ref{monotonicity} and~\ref{monotonicitynew1}
	 of App.~\ref{sec:mono}.
	   For the points in the white regions the composition rules cannot be used to determine a specific capacity ordering with respect to the capacity of the red dot element. } 		\label{figurnew3}
\end{figure}
The derivation of the second identity of Eq.~(\ref{IMPO12}) 
 follows along the same lines by simply  inverting the roles of $(\eta',N')$ and $(\eta,N)$ in the previous passages.
 \end{proof} 

Given  next  $g\geq 1$ and $N\geq 0$ 
  and the functions 
 \begin{eqnarray} \label{Ndefnattnamp}
	{N}_{g,N}^{(1)}(g')&:=& N\left(\tfrac{g-1}{g'-1}\right)\;, \\ 
			{N}_{g,N}^{{(2)}}(g')&:=&(N+1) \left(\tfrac{g-1}{g}\right)\left(\tfrac{g'}{g'-1}\right)-1\;.
	\label{Ndefnattnamp11}
	\end{eqnarray}
 We can show that

\begin{theorem}\label{th33}
The amplifier channel ${\cal A}_{g,N}$ admits
\begin{eqnarray}
\begin{cases}\label{regioneAamp} 
 {\mathbb{L}}^{({\rm amp},1)}_{g,N} &:=\left\{  {\scriptstyle{ (g',N'): g'\geq g \;,
N'\geq \max\{  {N}_{g,N}^{{(1)}}(g')}},0\} \right\}, \\
 {\mathbb{L}}^{({\rm amp},2)}_{g,N}  &:= \left\{  {\scriptstyle{ (g',N'): g'\leq g \;,
N'\geq   \max\{ {N}_{g,N}^{(2)}(g')}},0\} \right\},  \\ 
	 {\mathbb{L}}^{({\rm amp})}_{g,N}  &:= {\mathbb{L}}^{({\rm amp},1)}_{g,N}  \cup  {\mathbb{L}}^{({\rm amp},2)}_{g,N} \;, 
	\end{cases} 
	\end{eqnarray}  
	(light blue area on the right-hand-side of the  bottom panel of Fig.~\ref{figurnew3}),
	as subset of the corresponding two-element concatenation, low-ground capacity region 
	${\mathbb{L}}^{(0)}({\cal A}_{\eta,N})$,
	and 
	\begin{eqnarray}
\begin{cases} \label{regioneBamp} 
 {\mathbb{H}}^{({\rm amp},1)}_{g,N} &:=\left\{  {\scriptstyle{ (g',N'): 1 \leq g'\leq g\;,
0\leq N'\leq  {N}_{g,N}^{{(1)}}(g')}}\right\}, \\
 {\mathbb{H}}^{({\rm amp},2)}_{g,N}  &:= \left\{  {\scriptstyle{ (g',N'): g'\geq g \;,
0\leq N'\leq   {N}_{g,N}^{(2)}(g')}}\right\},  \\ 
	 {\mathbb{H}}^{({\rm amp})}_{g,N}  &:= {\mathbb{H}}^{({\rm amp},1)}_{g,N}  \cup {\mathbb{H}}^{({\rm amp},2)}_{g,N} \;, 
	\end{cases} 
	\end{eqnarray}  
	(yellow area on the~rigth-hand-side of bottom  panel of Fig.~\ref{figurnew3}), as
subspace of  ${\mathbb{H}}^{(0)}({\cal A}_{\eta,N})$, i.e.
	\begin{eqnarray} {\mathbb{L}}^{({\rm amp})}_{g,N}   \subseteq {\mathbb{L}}^{(0)}({\cal A}_{g,N})\;, \qquad 
	{\mathbb{H}}^{({\rm amp})}_{g,N} \subseteq {\mathbb{H}}^{(0)}({\cal A}_{g,N})\label{IMPO12amp}\;.\end{eqnarray}
\end{theorem}
\begin{proof}
To derive the first inclusion 
  we set $(g_3,N_3)= (g',N')$ and $(g_2,N_2)= (g,N)$ in Eq.~$\mathbf{(C_{2})}$ of Tab.~\ref{tab1}  to observe 
 for all $g_1\geq 1$ and $N_1\geq 0$, 
 \begin{eqnarray} {\cal A}_{g',N'}={\cal A}_{g,N}\circ{\cal A}_{g_1,N_1}\label{dec3} \end{eqnarray}  
is also a thermal channel with parameters 
\begin{eqnarray}
\begin{cases} g'= g_1 g \geq g \;, \\\\
N'=\tfrac{(g_1-1)g N_1 + (g-1)N}{g'-1}\geq   {N}_{g,N}^{{(1)}}(g')\;,
\end{cases} 
\label{defN'} 
\end{eqnarray} 
which  span the entire area ${\mathbb{L}}^{({\rm amp},1)}_{g,N}$ , as the inequality is saturated whenever $N_1=0$, leading to
\begin{eqnarray}  {\mathbb{L}}^{({\rm amp},1)}_{g,N}   \subseteq {\mathbb{L}}^{(0)}({\cal A}_{g,N})\label{primaprimap4imq} \;.
\end{eqnarray}  
We next invoke~$\mathbf{(C_{3.2})}$ of Tab.~\ref{tab1}  to observe that 
\begin{eqnarray} {\cal A}_{g',N'}=\mathcal{E}_{\eta_2,N_2}\circ{\cal A}_{g,N}\;, \label{dec4} \end{eqnarray} 
is an amplifier too 
with parameters 
\begin{equation}
\begin{cases} 
&g'= \eta_2 g\leq g \;, \\\\
&N'=\tfrac{(1-\eta_2)(2N_2 +1) + \eta_2 (g-1)(2 N+1) - (g'-1)}{2(g'-1)} \geq 
{N}_{g,N}^{(2)}(g')\label{defN'mewnew} \;,
\end{cases} 
\end{equation} 
that cover the entire subset ${\mathbb{L}}^{({\rm amp},2)}_{g,N}$, as the inequality is saturated whenever $N_2=0$.  Hence we have 
\begin{eqnarray}  {\mathbb{L}}^{({\rm amp},1)}_{g,N}   \subseteq {\mathbb{L}}^{(0)}({\cal A}_{g,N})\label{primaprimap4imq} \;.
\end{eqnarray}  
which together with (\ref{primaprimap4imq}) gives us the thesis. 
The derivation of the second inclusion of Eq.~(\ref{IMPO12amp}) 
 follows along the same lines by simply  inverting the roles of $(g',N')$ and $(g,N)$ in the previous passages.
 \end{proof} 

It is finally possible to establish a partial ordering between  the capacities of attenuators and those of the amplifiers.
Specifically  given  $\eta\in [0,1]$ and $N\geq 0$    and the 
 functions 
 \begin{eqnarray} \label{defn3g} 
	{N}_{\eta,N}^{(3)}(g)&:=&N \left(\tfrac{1-\eta}{\eta}\right)\left(\tfrac{g }{g-1}\right)-1 \;, \\
		{N}_{\eta,N}^{(4)}(g)&:=&(N+1) \left(\tfrac{1-\eta}{g-1}\right)\;, \label{defn3gw} 
	\end{eqnarray}
it follows that
		\begin{theorem}\label{th3new}
	The attenuator map  ${\cal E}_{\eta,N}$ admits 
\begin{eqnarray} \label{defC+at} 
 {\mathbb{X}}^{({\rm att},+)}_{\eta,N} &:=&\left\{  {\scriptstyle{ (g,N'): g\geq 1, 0\leq N' \leq {N}_{\eta,N}^{(3)}(g)}} \right\},
 \\ \nonumber 
  {\mathbb{X}}^{({\rm att},-)}_{\eta,N} &:=&\left\{  {\scriptstyle{ (g,N'): g \geq 1, N' \geq \max\{ {N}_{\eta,N}^{(4)}(g)}},0\}
   \right\},
	\end{eqnarray}  
	(yellow and light blue areas on the right-hand-side of the top  panel of Fig.~\ref{figurnew3})
	as subsets of ${\mathbb{H}}^{(0)}({\cal E}_{\eta,N})$ and $ {\mathbb{L}}^{(0)}({\cal E}_{\eta,N})$ respectively, i.e. 
	\begin{eqnarray}  {\mathbb{X}}^{({\rm att},+)}_{\eta,N}  \subseteq {\mathbb{H}}^{(0)}({\cal E}_{\eta,N})\;, \qquad 
	{\mathbb{X}}^{({\rm att},-)}_{\eta,N}  \subseteq {\mathbb{L}}^{(0)}({\cal E}_{\eta,N})\label{IMPO12amp}\;.\end{eqnarray}
\end{theorem}
\begin{proof}
To prove the first inclusion   we use the decomposition rule~$\mathbf{(C_{3.1})}$ of Tab.~\ref{tab1}, which 
setting $(\eta_3,N_3)= (\eta,N)$, $(g_1,N_1)=(g,N')$, and arbitrary $\eta_2\in[0,1]$, $N_2\geq 0$, 
allows us to write
\begin{eqnarray} \label{dec5} 
{\cal E}_{\eta,N} = {\cal E}_{\eta_2,N_2} \circ  {\cal A}_{g,N'} \;, 
\end{eqnarray} 
 for all
$g\geq 1$ and $N'\leq {N}_{\eta,N}^{(3)}(g)$, i.e. for the entire set ${\mathbb{X}}^{({\rm att},+)}_{\eta,N}$. 
The second inclusion follows instead from Eq.~$\mathbf{(C_{3.2})}$ of Tab.~\ref{tab1}, which setting 
setting $(\eta_2,N_2)= (\eta,N)$, $(g_3,N_3)=(g,N')$, and arbitrary $g_1\geq 1$, $N_1\geq 0$ allows us to write 
\begin{eqnarray} \label{dec6} 
{\cal A}_{g,N'} = {\cal E}_{\eta,N} \circ {\cal A}_{g_1,N_1}   \;, 
\end{eqnarray}  for all
$g\geq 1$ and $N'\geq {N}_{\eta,N}^{(4)}(g)$,  i.e. for the entire set ${\mathbb{X}}^{({\rm att},-)}_{\eta,N}$. 
\end{proof}
In a similar way,
given
 \begin{eqnarray} \label{ddf} 
	{N}_{g,N}^{(3)}(\eta)&:=& N \left(\tfrac{g-1}{1-\eta}\right)-1\;, \\
		{N}_{g,N}^{(4)}(\eta)&:=& (N+1)
		\left(\tfrac{g-1}{g}\right)\left(\tfrac{
		\eta}{1-\eta}\right)  \;, 
		\label{ddf1} 
	\end{eqnarray}
we have that 

\begin{theorem}\label{th33new}
	The amplifier map ${\cal A}_{g,N}$ admits 
\begin{eqnarray}\label{defC+am} 
 {\mathbb{X}}^{({\rm amp},+)}_{g,N} &:=\left\{  {\scriptstyle{ (\eta,N'): \eta\in [0,1], 0\leq N' \leq {N}_{g,N}^{(3)}(\eta)}} \right\},
 \\\nonumber 
  {\mathbb{X}}^{({\rm amp},-)}_{g,N} &:=\left\{  {\scriptstyle{ (\eta,N'): g\geq 1, N' \geq \max\{ {N}_{g,N}^{(4)}(\eta)}},0\}  \right\},
	\end{eqnarray}  
	(yellow and light blue areas on the left-hand-side of the bottom  panel of Fig.~\ref{figurnew3})
as subsets of  ${\mathbb{H}}^{(0)}({\cal A}_{g,N})$ and $ {\mathbb{L}}^{(0)}({\cal A}_{g,N})$ respectively , i.e. 
	\begin{equation}  {\mathbb{X}}^{({\rm amp},+)}_{g,N}  \subseteq {\mathbb{H}}^{(0)}({\cal A}_{g,N})\;, \qquad 
	{\mathbb{X}}^{({\rm amp},-)}_{g,N}  \subseteq {\mathbb{L}}^{(0)}({\cal A}_{g,N})\label{IMPO12ddnewnew}\;.\end{equation}
\end{theorem}
\begin{proof} The above inclusions are just an alternative way to cast  the results of
Property~\ref{th3new}. For the sake of symmetry we provide however an independent derivation. The first 
can be proven by using the decomposition rule~$\mathbf{(C_{3.2})}$ of Tab.~\ref{tab1}, which 
setting $(g_3,N_3)= (g,N)$, $(\eta_2,N_2)=(\eta,N')$, and $g_1\geq 1$, $N_1\geq 0$, 
allows us to write 
\begin{eqnarray} \label{dec7} 
{\cal A}_{g,N} = {\cal E}_{\eta,N'} \circ {\cal A}_{g_1,N_1}\;, 
\end{eqnarray} 
 for all
$\eta\in[0,1]$ and $0\leq N'\leq {N}_{g,N}^{(3)}(\eta)$, i.e. for all the points of ${\mathbb{X}}^{({\rm amp},+)}_{g,N}$. 
The second inclusion of Eq.~(\ref{IMPO12ddnewnew}) 
follows instead  from Eq.~$\mathbf{(C_{3.1})}$ of Tab.~\ref{tab1}, which setting 
setting $(\eta_3,N_3)= (\eta,N')$, $(g_1,N_1)=(g,N)$, and   $\eta_2\in [0,1]$, $N_2\geq 0$, gives  us 
\begin{eqnarray} \label{dec8} 
{\cal E}_{\eta,N'}  = {\cal E}_{\eta_2,N_2}  \circ {\cal A}_{g,N}\;, 
\end{eqnarray} 
 for all
$\eta\in [0,1]$ and $N'\geq {N}_{g,N}^{(4)}(\eta)$, i.e. for all the points of ${\mathbb{X}}^{({\rm amp},-)}_{g,N}$. 
\end{proof}

Inclusions which are different 
with respect to the one reported in the Properties can be obtained by reversing
the order of the decompositions employed in the proofs. Such  inequalities however  are provably less performant than those presented. 
For instance setting $(\eta_3,N_3)= (\eta',N')$ and $(\eta_2,N_2)= (\eta,N)$
 in  Eq.~$\mathbf{(C_1)}$
allows us  to determine that $(\eta',N') \in  {\mathbb{L}}({\cal E}_{\eta,N})$
for all 
\begin{eqnarray} \eta'\leq \eta\;, \qquad N'\geq  \left(\frac{1-\eta}{1-\eta'}\right)N\;, \label{new} \end{eqnarray}  a
result which is implied by~(\ref{IMPO1}) since the set ${\mathbb{L}}^{({\rm att},2)}_{\eta,N}$ 
includes all points fulling the condition~(\ref{new}). 
		Similarly using Eq.~$\mathbf{(C_{4.1})}$ instead
		of~$\mathbf{(C_{3.1})}$ allows one to show
	 that i)  $(\eta',N') \in {\mathbb{L}}({\cal E}_{\eta,N})$
	 for all $\eta'\geq \eta$ and $N'\geq  	  \frac{\eta'-\eta+\eta'(1-\eta)N}{\eta(1-\eta')}$
	 (a condition that is already implied by (\ref{IMPO12})); ii) 
	 $(\eta',N') \in {\mathbb{H}}({\cal E}_{\eta,N})$
 for all 
	  $\eta'\leq \eta$ and $N'\leq  
	  \frac{\eta'-\eta+\eta'(1-\eta)N}{\eta(1-\eta')}$ (which again is implied by~(\ref{IMPO12})); iii) $(g,N')\in {\mathbb{H}}({\cal E}_{\eta,N})$
	 for all $g\geq 1$ and $N'\leq  	 N\left( \tfrac{1-\eta}{g-1}\right)  -1$ 
	  (implied by~the first inclusion of Eq.~(\ref{IMPO12amp})); {\it iv)} 
	  $(g,N')\in {\mathbb{L}}({\cal E}_{\eta,N})$
 for all $g\geq 1$ and $N'\geq  	(N+1) \left( \tfrac{g}{g-1}\right) \left( \tfrac{1-\eta}{\eta}\right)$ 
	  (implied by the second inclusion of  Eq.~(\ref{IMPO12amp})). 
Putting together these results  we can hence conclude that the two-elements concatenation regions 
${\mathbb{L}}^{(0)}({\cal E}_{\eta,N})$ and ${\mathbb{L}}^{(0)}({\cal A}_{g,N})$ coincide respectively with 
${\mathbb{L}}^{\rm{(att)}}_{\eta,N} \bigcup {\mathbb{X}}^{(\rm{att},-)}_{\eta,N}$ and  
${\mathbb{L}}^{\rm{(amp)}}_{g,N} \bigcup {\mathbb{X}}^{(\rm{amp},-)}_{g,N}$, a condition which  translated into the parametrization~(\ref{defgenerale})
can be expressed as 
  \begin{eqnarray} \label{LOW}
{\mathbb{L}}^{(0)}_{x,M}&=& \left\{ \begin{array}{ll} {\mathbb{L}}^{\rm{(att)}}_{x,M/(1-x)} \bigcup {\mathbb{X}}^{(\rm{att},-)}_{x,M/(1-x)} & \mbox{for 
$x\in [0,1]$} \;,  \\ \\
{\mathbb{L}}^{\rm{(amp)}}_{x,M/(x-1)} \bigcup {\mathbb{X}}^{(\rm{amp},-)}_{x,M/(x-1)} & \mbox{for 
$x\geq  1$}\;.
\end{array} \right. 
\end{eqnarray} 
Analogously we have that ${\mathbb{H}}^{(0)}({\cal E}_{\eta,N})$ and 
${\mathbb{H}}^{(0)}({\cal A}_{g,N})$ correspond respectively to  
${\mathbb{H}}^{\rm{(att)}}_{\eta,N} \bigcup {\mathbb{X}}^{(\rm{att},+)}_{\eta,N}$ and $
{\mathbb{H}}^{\rm{(amp)}}_{g,N} \bigcup {\mathbb{X}}^{(\rm{amp},+)}_{g,N}$, so that 
\begin{eqnarray} 
{\mathbb{H}}^{(0)}_{x,M}&=& \left\{ \begin{array}{ll} {\mathbb{H}}^{\rm{(att)}}_{x,M/(1-x)}  \bigcup {\mathbb{X}}^{(\rm{att},+)}_{x,M/(1-x)} & \mbox{for 
$x\in [0,1]$} \;,  \\ \\
{\mathbb{H}}^{\rm{(amp)}}_{x,M/(x-1)}  \bigcup {\mathbb{X}}^{(\rm{amp},+)}_{x,M/(x-1)} & \mbox{for 
$x\geq  1$}\;. 
\end{array} \right.\label{HIGH} 
\end{eqnarray} 
The border lines of these regions are provided by the curves ${M}_{x,M}^{{(j)}}(x')$ of Tab.~\ref{tab2}
obtained from ${N}_{\eta,N}^{(1,2)}(\eta')$, ${N}_{g,N}^{(1,2)}(g')$,
${N}_{\eta,N}^{(3,4)}(g')$, and ${N}_{g,N}^{(1,2)}(\eta')$ via the substitutions 
\begin{eqnarray} 
{M}_{x,M}^{{(j)}}(x') := |1-x'| {N}_{x,M/|1-x|}^{(j)}(x')\;. 
\end{eqnarray} 
For instance one has that 
${\mathbb{L}}^{(0)}_{x,M}$ is given by  the points $(x',M')$ such that 
\begin{equation} \label{comp1} 
 M' \geq \left\{ \begin{array}{lr} {M}_{x,M}^{{(1)}}(x')= M\frac{x'}{x},  &  \quad  (0\leq x'\leq x), \\
{M}_{x,M}^{{(2)}}(x')= M-x +x', &\quad   (x\leq x'\leq 1), \\
{M}_{x,M}^{{(4)}}(x')= M+1-x ,  & \quad   (1\leq  x'), \\
\end{array} \right.\end{equation} 
for  $x\leq 1$, and 
\begin{equation}\label{comp2} 
 M' \geq \left\{ \begin{array}{lr} {M}_{x,M}^{{(4)}}(x')= (M+x-1)\frac{x'}{x},  &
   (0\leq x'\leq 1), \\
{M}_{x,M}^{{(2)}}(x')= (M-1)\frac{x'}{x} +1,  &  (1\leq x'\leq x), \\
{M}_{x,M}^{{(1)}}(x')=M,  &   (x\leq  x'), \\
\end{array} \right.
\end{equation} 
for $x\geq 1$. 
Viceversa we have that 
${\mathbb{H}}^{(0)}_{x,M}$ includes all points $(x',M')$ such that 
\begin{equation}\label{comp3} 
0\leq M' \leq \left\{ \begin{array}{lr} {M}_{x,M}^{{(2)}}(x')= M-x +x',  &  (0\leq x'\leq x), \\
{M}_{x,M}^{{(1)}}(x')= M\frac{x'}{x}, &  (x\leq x'\leq 1), \\
{M}_{x,M}^{{(3)}}(x')= (M-x) \frac{x'}{x}+1,  &  (1\leq  x'), \\
\end{array} \right.\end{equation} 
for  $x\leq 1$, and 
\begin{equation}\label{comp4} 
0\leq M' \leq \left\{ \begin{array}{lr} {M}_{x,M}^{{(3)}}(x')=M-1 +x',  &  (0\leq x'\leq 1), \\
{M}_{x,M}^{{(1)}}(x')= M,  &  (1\leq x'\leq x), \\
{M}_{x,M}^{{(2)}}(x')=(M-1)\frac{x'}{x} +1,  &  (x\leq  x'), \\
\end{array} \right.
\end{equation} 
for $x\geq 1$. 
		\begin{center} 		
			\begin{table*}[t!]
\begin{tabular}{|lllr|}
\hline
${N}_{\eta,N}^{{(1)}}(\eta')$ &$\mapsto$& 
$ {M}_{x,M}^{{(1)}}(x'):= (1-\eta'){N}_{\eta,N}^{{(1)}}(\eta')\Big|_{\substack{\scriptstyle{\eta=x,\eta'=x'}\\
\scriptstyle{N=M/(1-x)}}} = M {x' }/x\;,
$&
  \\ &&& for $x,x'\in[0,1]$  \\
${N}_{\eta,N}^{{(2)}}(\eta')$ &$\mapsto$& 
$ {M}_{x,M}^{{(2)}}(x'):=(1-\eta'){N}_{\eta,N}^{{(2)}}(\eta')\Big|_{\substack{\scriptstyle{\eta=x,\eta'=x'}\\
\scriptstyle{N=M/(1-x)}}} =  M-x +x' \;,
$& 
  \\ \hline
  ${N}_{g,N}^{{(1)}}(g')$ &$\mapsto$& 
$ {M}_{x,M}^{{(1)}}(x'):=(g'-1){N}_{g,N}^{{(1)}}(g')\Big|_{\substack{\scriptstyle{g=x,g'=x'}\\
\scriptstyle{N=M/(x-1)}}}  =  M\;,$& 
  \\   &&&for $x,x'\geq 1$\\
${N}_{g,N}^{(2)}(g')$ &$\mapsto$& 
$ {M}_{x,M}^{(2)}(x'):= (g'-1){N}_{g,N}^{{(2)}}(g')\Big|_{\substack{\scriptstyle{g=x,g'=x'}\\
\scriptstyle{N=M/(x-1)}}} = 
(M-1)x'/x +1\;,
$& 
  \\ \hline
    ${N}_{\eta,N}^{(3)}(g')$ &$\mapsto$& 
$ {M}_{x,M}^{(3)}(x'):=(g'-1){N}_{\eta,N}^{(3)}(g')\Big|_{\substack{\scriptstyle{\eta=x,g'=x'}\\
\scriptstyle{N=M/(1-x)}}}= 
(M-x) x'/x+1\;,$& 
 \\ &&&   for $x\in[0,1]$  and $x'\geq 1$ \\ 
${N}_{\eta,N}^{(4)}(g')$ &$\mapsto$& 
$ {M}_{x,M}^{(4)}(x'):= (g'-1) {N}_{\eta,N}^{(4)}(g') \Big|_{\substack{\scriptstyle{\eta=x,g'=x'}\\
\scriptstyle{N=M/(1-x)}}}=
M+ 1-x\;,
$&  \\ \hline
      ${N}_{g,N}^{(3)}(\eta')$ &$\mapsto$& 
${M}_{x,M}^{(3)}(x'):=(1-\eta') {N}_{g,N}^{(3)}(\eta')  \Big|_{\substack{\scriptstyle{g=x,\eta'=x'}\\
\scriptstyle{N=M/(x-1)}}}=
 M-1+x'\;,$& \\
&&&   for $x\geq 1$ and $x'\in[0,1]$ 
  \\ ${N}_{g,N}^{(4)}(\eta')$ &$\mapsto$& 
$ {M}_{x,M}^{(4)}(x'):=(1-\eta') {N}_{g,N}^{(4)}(\eta')\Big|_{\substack{\scriptstyle{g=x,\eta'=x'}\\
\scriptstyle{N=M/(x-1)}}}= 
 (M+x-1)x'/x\;,$& 
  \\ \hline
\end{tabular}
				\caption{Border lines of the regions ${\mathbb{L}}^{(0)}_{x,M}$  of Eq.~(\ref{LOW}) and ${\mathbb{H}}^{(0)}_{x,M}$ of Eq.~(\ref{HIGH})  
				for the $\Phi_{x,M}$ map. 
				\label{tab2}} 
			\end{table*}
		\end{center}

\subsection{Three-elements concatenation regions for  channels $\Phi_{x,M}$ fulfilling Eq.~(\ref{notEB})}~\label{three-times} 
Comparing the l.h.s. of  Eqs.~(\ref{comp1})--(\ref{comp4}) with the functions
$f^{(1)}_{x,M}(x')$ and $f^{(1)}_{x,M}(x') $ of Eq.~(\ref{effe2}), one can easily check 
 that  for channels $\Phi_{x,M}$ which  fulfil Eq.~(\ref{notEB}), the  two-element concatenation regions 
 ${\mathbb{L}}^{(0)}_{x,M}$ and ${\mathbb{H}}^{(0)}_{x,M}$ can be expressed as 
   \begin{eqnarray} \label{analytical10} 
  {\mathbb{L}}^{(0)}_{x,M} &=&  
  \Big\{ (x',M')\in\left(\mathbb{R}^{+} \right)^{2} :  
  \\ \nonumber && \quad  M'\geq \max\{  f^{(1)}_{x,M}(x'),  f^{(2)}_{x,M}(x')\} \Big\} \;,
\\ \label{analytical20} 
  {\mathbb{H}}^{(0)}_{x,M} &=&  
  \Big\{ (x',M')\in\left(\mathbb{R}^{+} \right)^{2} :  
  \\ \nonumber && \quad  M'\leq \min\{  f^{(1)}_{x,M}(x'),  f^{(2)}_{x,M}(x')\} \Big\} \;.
  \end{eqnarray}
Accordingly under the condition~(\ref{notEB}), the proof  of the identities~(\ref{analytical1}) and (\ref{analytical2})  reduces hence to show that  
the two-element concatenation regions ${\mathbb{L}}^{(0)}_{x,M}$ and ${\mathbb{H}}^{(0)}_{x,M}$ correspond to the three-elements concatenations sets
 ${\mathbb{L}}_{x,M}$ and ${\mathbb{H}}_{x,M}$, i.e. 
 \begin{corollary}\label{inclusionsrules1}
Given $(x,M)$ such that $M\leq M_{\text{EB}}(x)$ we have that 
\begin{eqnarray}\label{ide1} 
 {\mathbb{L}}^{(0)}_{x,M} =  {\mathbb{L}}_{x,M}\;, \qquad   {\mathbb{H}}^{(0)}_{x,M} =  {\mathbb{H}}_{x,M}\;. \end{eqnarray}
\end{corollary}

\begin{proof} 
 This result  can be derived by noticing that for channels $\Phi_{x,M}$ which are non-EB, ${\mathbb{L}}^{(0)}_{x,M}$ and ${\mathbb{H}}^{(0)}_{x,M}$
 fulfil the same ordering rules of ${\mathbb{L}}_{x,M}$ and ${\mathbb{H}}_{x,M}$ given in 
 Eqs.~(\ref{Hide11}) and (\ref{Hide22}), i.e. 
\begin{equation}
M < M_{\text{EB}}(x)  \Longrightarrow \left\{ 
\begin{array}{l} 
 {\mathbb{L}}^{(0)}_{x',M'} \subseteq  {\mathbb{L}}^{(0)}_{x,M} \;,  \forall (x',M') \in {\mathbb{L}}^{(0)}_{x,M} 
 \\\\
 \label{ide22} 
 {\mathbb{H}}^{(0)}_{x',M'} \subseteq  {\mathbb{H}}^{(0)}_{x,M} \;,  \forall (x',M') \in {\mathbb{H}}^{(0)}_{x,M} 
 \end{array} \right.
\end{equation}  
The derivation of these  identities   relies on
 a series geometric relations in which one has to compare the relative size and position of the
 two-dimensional polytopes defined in Eqs.~(\ref{analytical10}) and (\ref{analytical20}). It suffices to do show this for the high ground region. In fact, since $(x',M')\in {\mathbb{L}}^{(0)}_{x,M} $ if and only if $(x,M)\in {\mathbb{H}}^{(0)}_{x',M'}$, we have that $(x'',M'')\in {\mathbb{L}}^{(0)}_{x,M}$, $(x'',M'')\in {\mathbb{L}}^{(0)}_{x',M'}$, $(x',M')\in {\mathbb{L}}^{(0)}_{x,M}$ if and only if $(x,M)\in {\mathbb{H}}^{(0)}_{x'',M''}$, $(x',M')\in {\mathbb{H}}^{(0)}_{x'',M''}$, $(x,M)\in {\mathbb{H}}^{(0)}_{x',M'}$. Therefore it suffices to prove only that the latter is true whenever  $(x',M')\in {\mathbb{H}}^{(0)}_{x'',M''}$ and $(x,M)\in {\mathbb{H}}^{(0)}_{x',M'}$. This can be proven by case-by-case inspection, and we outline the derivation with the aid of Fig.~\ref{figure9inclusion}. 
\begin{figure}
	\includegraphics[width=\columnwidth]{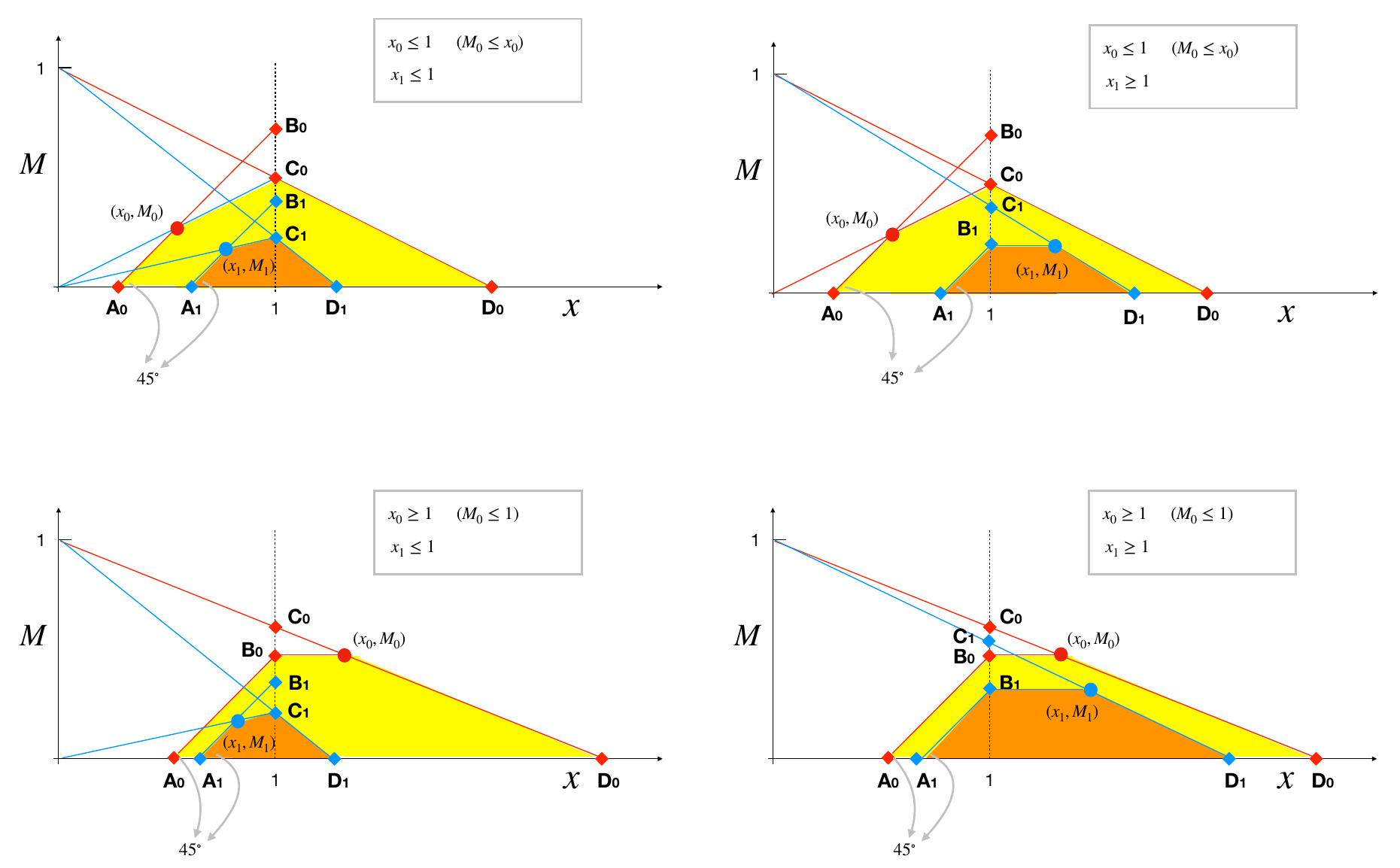}
	\caption{Graphical explanation of inclusions rules of Eq.~(\ref{ide22}). The border of the high ground region can have two different shapes, depending on $x_i$ being smaller or larger than 1 (the degenerate case $x_i=1$ is not plotted but it is analogous). In the case $x_i<1$, the region is obtained as the intersection of $\left(\mathbb{R^+}\right)^2$ with the half-planes delimited by a $45$ degrees line passing through $(x_i,M_i)$, a line passing through the origin and $(x_i,M_i)$ and intersecting the $x=1$ line at $\mathbf{C_i}$, and a line passing through $\mathbf{C_i}$ and $(0,1)$. In the case $x_i>1$, the region is obtained as the intersection of $\left(\mathbb{R^+}\right)^2$ with the half-planes delimited by an horizontal line passing through $(x_i,M_i)$ and intersecting the $x=1$ line at $\mathbf{B_i}$, a $45$ degrees line intersecting $\mathbf{B_i}$, and a line passing through $(x_i,M_i)$ and $(0,1)$.
	If $x_0<1$ and $x_1<1$, all the the intersection of $\left(\mathbb{R^+}\right)^2$  and the half-planes individuated by the segments $\mathbf{A_1}-(x,M)$, $(x,M)-\mathbf{C_1}$ and $\mathbf{C_1}-\mathbf{D_1}$  are contained in the corresponding region individuated by $x_0$, and therefore their intersection also satisfy the same inclusions. If $x_0>1$ and $x_1>1$, all the the intersection of $\left(\mathbb{R^+}\right)^2$  and the half-planes individuated by the segments $\mathbf{A_1}-\mathbf{B_1}$, $\mathbf{B_1}-(x,M)$,and $(x,M)-\mathbf{D_1}$ are contained in the corresponding region individuated by $x_0$, and therefore their intersection also satisfy the same inclusions. 
	If $x_0<1$ and $x_1>1$, a calculation shows that in the non EB region the region on the left is convex, therefore it contains the triangle $\mathbf{A_1}-\mathbf{B_1}-(1,0)$, while the inclusion of the triangle $\mathbf{C_1}-\mathbf{D_1}-(1,0)$ follows from the same half-plane argument as before. If $x_0>1$ and $x_1<1$, a calculation shows that in the non EB region the region on the right is convex, therefore it contains the triangle $\mathbf{C_1}-\mathbf{D_1}-(1,0)$, while the inclusion of the triangle $\mathbf{A_1}-\mathbf{C_1}-(1,0)$ follows from the same half-plane argument as before.}
	\label{figure9inclusion}
\end{figure}
Take hence $(x',M')\in  {\mathbb{L}}_{x,M}$: from~(\ref{alternativedefLO}) it follows that we can write
\begin{eqnarray} 
\Phi_{x',M'} = \Phi_{\bar{x}_1,\bar{M}_1} \circ  \Phi_{x,M} \circ \Phi_{\bar{x}_2,\bar{M}_2}\;,
\end{eqnarray} 
for some proper choice of $(\bar{x}_1,\bar{M}_1),  (\bar{x}_2,\bar{M}_2) \in \left(\mathbb{R}^{+} \right)^{2}$.
Setting then $\Phi_{x_2,M_2}  := \Phi_{x,M} \circ \Phi_{\bar{x}_2,\bar{M}_2}$, we can claim that $({x_2,M_2})$ is an element of ${\mathbb{L}}^{(0)}_{x,M}$ and  $({x',M'})$ an element of ${\mathbb{L}}^{(0)}_{x_2,M_2}$ (indeed
$\Phi_{x',M'} = \Phi_{\bar{x}_1,\bar{M}_1} \circ \Phi_{x_2,M_2}$).
 However, because of Eq.~(\ref{ide22}) the latter is a subset of ${\mathbb{L}}^{(0)}_{x,M}$, so  we can conclude that 
 \begin{eqnarray}
(x',M')\in  {\mathbb{L}}_{x,M} \Longrightarrow 
(x',M') \in {\mathbb{L}}^{(0)}_{x,M} \;,
\end{eqnarray} 
which together with Eq.~(\ref{impoinclusions}) gives the first of the  identities Eq.~(\ref{ide1}).

By the same token, let $(x',M')\in  {\mathbb{H}}_{x,M}$:  from~(\ref{alternativedefHIGH}) it follows that we can write
\begin{eqnarray} 
\Phi_{x,M} = \Phi_{\bar{x}_1,\bar{M}_1} \circ  \Phi_{x',M'} \circ \Phi_{\bar{x}_2,\bar{M}_2}\;,
\end{eqnarray} 
for some proper choice of $(\bar{x}_1,\bar{M}_1),  (\bar{x}_2,\bar{M}_2) \in \left(\mathbb{R}^{+} \right)^{2}$.
Setting then $\Phi_{x_2,M_2}  := \Phi_{x',M'} \circ \Phi_{\bar{x}_2,\bar{M}_2}$, we can claim that 
$({x',M'})$ is an element of ${\mathbb{H}}^{(0)}_{x_2,M_2}$, and $({{x}_2,{M}_2})$ an element of ${\mathbb{H}}^{(0)}_{x,M}$ 
(indeed $\Phi_{x,M} = \Phi_{\bar{x}_1,\bar{M}_1} \circ  \Phi_{x_2,M_2}$).
 However, because of Eq.~(\ref{ide22}) it follows that ${\mathbb{H}}^{(0)}_{x_2,M_2}$ is included into 
  ${\mathbb{H}}^{(0)}_{x,M}$, so that 
 \begin{eqnarray}
(x',M')\in  {\mathbb{H}}_{x,M} \Longrightarrow 
(x',M') \in {\mathbb{H}}^{(0)}_{x,M} \;,
\end{eqnarray} 
which gives the second of the identities Eq.~(\ref{ide1}).
\end{proof} 

 \subsection{Three-elements concatenation regions for channels $\Phi_{x,M}$ fulfilling~(\ref{EBconstraint})}~\label{three-timesEB} 
 Here we show that for channels $\Phi_{x,M}$ which are deep in the  EB region (i.e. such that (\ref{EBconstraint}) holds true), 
 the low-ground/high-ground regions are determined by Eqs.~(\ref{analytical2new}) and ~(\ref{analytical2new11}).
 To begin with let us observe that, for $M\geq M_{\text{EB}}(x)$, 
   Eqs.~(\ref{comp1})--(\ref{comp4}) lead to express the two-elements concatenation sets ${\mathbb{L}}^{(0)}_{x,M}$ and 
 ${\mathbb{H}}^{(0)}_{x,M}$ as
      \begin{eqnarray} \label{analytical10} 
  {\mathbb{L}}^{(0)}_{x,M} &=&  
  \Big\{ (x',M')\in\left(\mathbb{R}^{+} \right)^{2} :  
  \\ \nonumber && \quad  M'\geq \min\{  f^{(1)}_{x,M}(x'),  f^{(2)}_{x,M}(x')\} \Big\} \;,
  \\ \label{analytical20} 
  {\mathbb{H}}^{(0)}_{x,M} &=&  
  \Big\{ (x',M')\in\left(\mathbb{R}^{+} \right)^{2} :  
  \\ \nonumber && \quad  M'\leq \max\{  f^{(1)}_{x,M}(x'),  f^{(2)}_{x,M}(x')\} \Big\} \;,
  \end{eqnarray}
with $f^{(1,2)}_{x,M}(x')$ the functions defined in Eq.~(\ref{effe2}).  It turns out that such regions always admit
a non trivial overlap  ${\mathbb{O}}^{(0)}_{x,M} := {\mathbb{L}}^{(0)}_{x,M} \bigcap {\mathbb{H}}^{(0)}_{x,M}$ that includes a finite portion of the plane $\left(\mathbb{R}^{+} \right)^{2}$  in the neighbourhood of the origin  -- see Fig.~\ref{figure8inclusion1}.
\begin{figure}
	\includegraphics[width=\columnwidth]{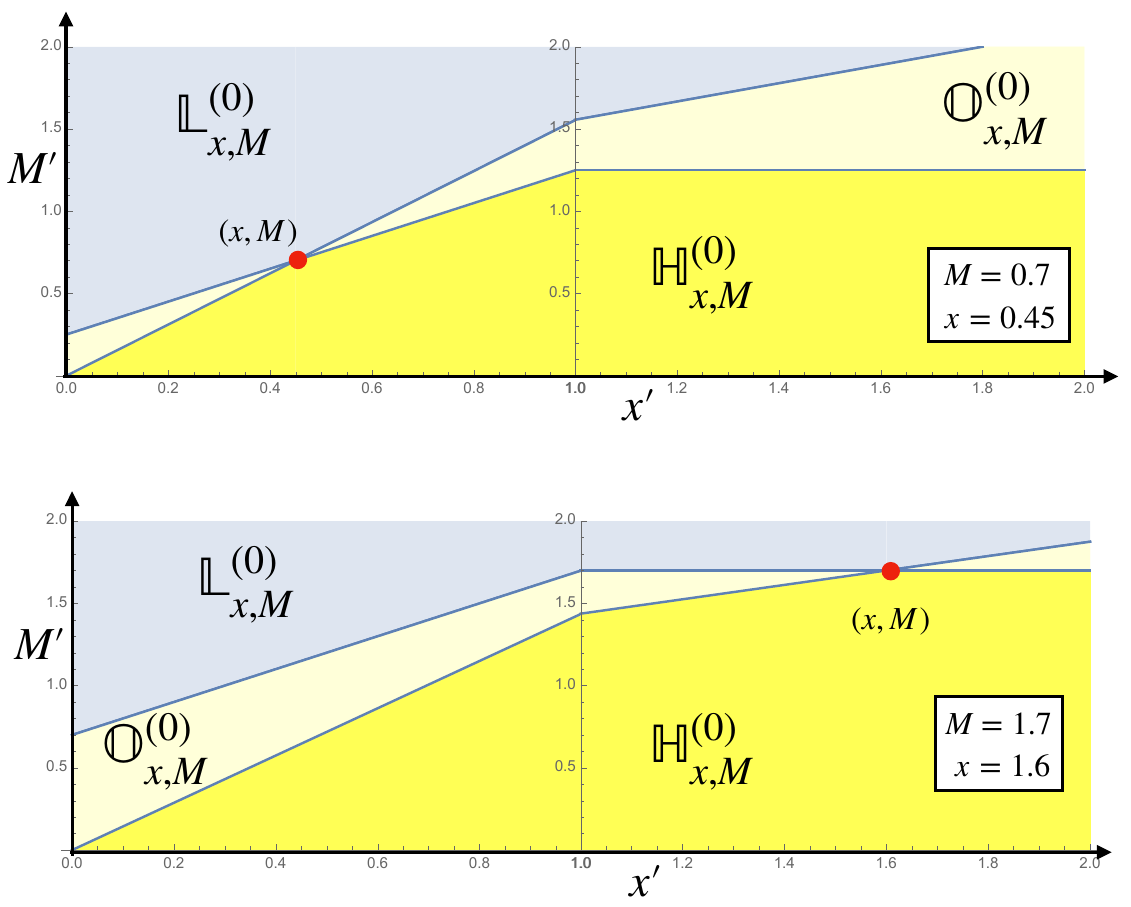}
	\caption{Examples of the two-elements compositions low-ground/high-ground regions
	${\mathbb{L}}^{(0)}_{x,M}$ and ${\mathbb{H}}^{(0)}_{x,M}$ for EB channels $\Phi_{x,M}$.
	Notice that such sets admits a non trivial overlap ${\mathbb{O}}^{(0)}_{x,M} := {\mathbb{L}}^{(0)}_{x,M} \bigcap {\mathbb{H}}^{(0)}_{x,M}$ (pale yellow region) which always includes the origin point $(0,0)$. 
	 }
	\label{figure8inclusion1}
\end{figure}
As a matter of fact  $(0,0)$ itself can be considered as element of ${\mathbb{O}}^{(0)}_{x,M}$  {(to be precise, it is an element of the closure of ${\mathbb{O}}^{(0)}_{x,M}$ -- see App.~\ref{app:formaL}
 for a refinement of the following argument, taking this into account)}, 
i.e. 
 \begin{equation} \label{notalimite} 
 M > M_{\text{EB}}(x)  \Longrightarrow (0,0) \in {\mathbb{O}}^{(0)}_{x,M} \subseteq {\mathbb{H}}^{(0)}_{x,M}\subseteq
 {\mathbb{H}}_{x,M}
 \;.
 \end{equation} 
 Now from~(\ref{trivial}) and (\ref{complementary}) we know that 
 \begin{eqnarray} \label{inclusion0} 
  (x',M') \in {\mathbb{H}}_{0,0} \;, \qquad \forall (x',M') \in \left(\mathbb{R}^{+} \right)^{2} \;, 
 \end{eqnarray} 
which using~(\ref{Hide11}) gives
\begin{equation} \label{ddfs} 
(x',M') \in {\mathbb{H}}_{x,M} \;, \qquad \forall (x',M') \in \left(\mathbb{R}^{+} \right)^{2} \;,
\end{equation} 
hence proving 
 Eq.~(\ref{analytical2new11})
Observe next that if $(x',M')$ is also EB, then from (\ref{ddfs})
we also have 
$(x,M)\in  {\mathbb{H}}_{x',M'}$ which lead to  Eq.~(\ref{analytical2new})  via the complementarity rule~(\ref{complementary}).

\section{Stabilizing  bounds} \label{sec:stab}

Building up from the low-ground/high-ground analysis presented in the previous sections we can now detail a general technique that 
allows one to potentially improve existing upper and lower bounds for the capacities of single-mode PI-GBCs. Indeed suppose that   ${\cal K}^{(+)}( {x,M})$, ${\cal K}^{(-)}( {x,M})$ are two functions  which bound for the capacity ${\cal K}({x,M}):= {\cal K}( \Phi_{x,M})$ of the channel $\Phi_{x,M}$, i.e.
\begin{equation} \label{supp} 
{\cal K}^{(+)}( {x,M}) \geq {\cal K}( {x,M}) \geq {\cal K}^{(+)}( {x,M}) \;, 
\end{equation} 
for all $(x,M) \in \left(\mathbb{R}^{+} \right)^{2}$.
From~(\ref{prop1}) and (\ref{prop2})  it follows that the quantities 
\begin{eqnarray}  \label{improatt} 
\begin{cases}\label{improvboun}
\underline{\cal K}^{(+)}({x,M}) &:= \min_{({x',M'}) \in  {\mathbb{H}}_{x,M} }
 {\cal K}^{(+)}({x',M'})   \;,\\\\ 
 \overline{\cal K}^{(-)}({x,M}) &:= \max_{({x',M'}) \in  {\mathbb{L}}_{x,M} } 
 {\cal K}^{(-)}({x',M'})\;, \end{cases} 
\end{eqnarray} 
can  potentially improve the inequalities~(\ref{supp}), i.e. 
\begin{eqnarray} \label{improvebound} 
\begin{cases}
&{{\cal K}}^{(+)}( {x,M}) \geq \underline{\cal K}^{(+)}({x,M}) \;,
\\  \\
  &\underline{\cal K}^{(+)}( {x,M}) \geq {\cal K}( {x,M}) \geq \overline{{\cal K}}^{(-)}
( {x,M})\;,
 \\ \\
&\overline{{\cal K}}^{(-)}( {x,M})\geq {{\cal K}}^{(-)}( {x,M})\;. 
\end{cases} 
\end{eqnarray} 
Of course it is very possible that  the functions $\underline{\cal K}^{(+)}({x,M})$ and $\overline{{\cal K}}^{(-)}( {x,M})$
will coincides with ${\cal K}^{(+)}({x,M})$ and ${{\cal K}}^{(-)}( {x,M})$, respectively: this happens for instance for all bounding functions~(\ref{supp}) which arise from
 operational procedures that automatically incorporate 
 data processing, e.g. the bounds~(\ref{UPPLOB1}), (\ref{lower1}), and 
 (\ref{lower2}). 
There are however  examples where the construction  
(\ref{improvboun}) leads to non trivial overall improvements. 
In what follow we shall detail one of such cases.

\subsection{Improving the upper bounds for the quantum capacity of thermal attenuators} 

Expressing the functions~(\ref{DEFFKBatt}) in terms of the parametrization~(\ref{defgenerale}) 
we can claim that 
 the quantum and private capacities of the channel $\Phi_{x,M}$ is always smaller than or equal to 
 \begin{equation} \label{new00} 
Q_\mathrm{FKG}(x,M) :=\left\{ \begin{array}{ll}
 Q^{\mathrm{att}}_\mathrm{FKG}(x,\tfrac{M}{1-x})\;, & \forall x\in [0,1]\;, \\ \\ 
Q^{\mathrm{amp}}_\mathrm{FKG}(M)\;, & \forall x\geq 1 \;. 
 \end{array} \right.
\end{equation}
Following  Eq.~(\ref{improvboun}) we can produce {an upper bound} $\underline{Q}_\mathrm{FKG}(x,M)$ by
taking the minimum of $\underline{Q}_\mathrm{FKG}(x',M')$ over the set ${\mathbb{H}}_{x,M}$. This strategy had been suggested and shown to give improvements in~\cite{flagged channel 3}, and it will be fully explored here. Without loss of generality we 
assume $(x,M)$ not to belong to the AD domain~${\mathbb{A}}$, i.e. 
\begin{eqnarray} \label{conc} 
\begin{cases} 
x\geq 1/2, \\
 M\leq M_{\text{AD}}(x)=  \min\{ x-1/2,1/2\}\;,
\end{cases} 
\end{eqnarray} 
a condition which via   Eq.~(\ref{analytical2}), allows us to identify 
 the high-ground set of the model with the
yellow regions of~Fig.~\ref{figurnew3bis}.
We next 
  compute the value $\underline{Q}^{(1)}_\mathrm{FKG}(x,M)$ which represents the minimum of ${Q}_\mathrm{FKG}(x',M')$ 
 for points of ${\mathbb{H}}_{x,M}$ which have $x'\in [0,1]$, and the value $\underline{Q}^{(2)}_\mathrm{FKG}(x,M)$
which instead involves points of ${\mathbb{H}}_{x,M}$ with { $x\geq1$}, i.e. \begin{eqnarray} \nonumber \label{bounds with opt}
\underline{Q}^{(1)}_\mathrm{FKG}(x,M)&:=& \min_{(x',M')\in {\mathbb{H}}_{x,M}; x'\in[0, 1]}
Q^{\mathrm{att}}_\mathrm{FKG}(x',\tfrac{M'}{1-x'})\;, \\
\underline{Q}^{(2)}_\mathrm{FKG}(x,M)&:=& \min_{(x',M')\in {\mathbb{H}}_{x,M}; x'\geq 1}
Q^{\mathrm{amp}}_\mathrm{FKG}(M')\;.
\end{eqnarray} 
Once we have these terms we can then write the global minimum of
 ${Q}_\mathrm{FKG}(x,M)$ over ${\mathbb{H}}_{x,M}$ as
\begin{equation} \label{asdf111} 
\underline{Q}_\mathrm{FKG}(x,M) = \min\{ \underline{Q}^{(1)}_\mathrm{FKG}(x,M),\underline{Q}^{(2)}_\mathrm{FKG}(x,M)\} \;. 
\end{equation}
Consider first the evaluation of $\underline{Q}^{(2)}_\mathrm{FKG}(x,M)$. Let us start observing that from Eqs.~(\ref{comp3}) and (\ref{comp4}), the maximum value  of $M'$ we can get  
for points $(x',M')$ of  ${\mathbb{H}}_{x,M}$ with $x'\geq 1$ is 
\begin{equation} \label{new3} 
M^{(>)}_{\max}(x) : =\left\{ \begin{array}{ll} 
M/x\;, & \forall x\in [\tfrac{1}{2},1]\\
M \;, & \forall x\geq 1 \end{array} \right. = \frac{M}{M_{\text{EB}}(x)}\;, 
\end{equation}
which, in virtue of Eq.~(\ref{conc}) is always smaller than or equal to $1/2$.
Recalling hence that on the interval $\kappa\in [0,1/2]$ the function
 $Q^{\mathrm{amp}}_\mathrm{FKG}(\kappa)$  is monotonically decreasing,
we can write 
\begin{eqnarray}\label{new2} 
&&\underline{Q}^{(2)}_\mathrm{FKG}(x,M) = Q^{\mathrm{amp}}_\mathrm{FKG}(M^{(>)}_{\max}(x))
 \\&&=-\log_2( \tfrac{eM}{M_{\text{EB}}(x)})+2h \left(\tfrac{\sqrt{ M^2+M^2_{\text{EB}}(x)}-M^2_{\text{EB}}(x)}{2 M_{\text{EB}}(x)}\right).\nonumber
 \end{eqnarray} 
 In the case of amplifiers (i.e. for $x \geq 1$) this implies that  $\underline{Q}^{(2)}_\mathrm{FKG}(x,M)$ always coincides with the old bound (\ref{new00}), so
 no improvement can be obtained. 
 \\
 
Consider next $\underline{Q}^{(1)}_\mathrm{FKG}(x,M)$. Here the key observation is  that 
for fixed value of $x'$, $Q^{\mathrm{att}}_\mathrm{FKG}(x',\tfrac{M'}{1-x'})$ is a decreasing function of $M'$.
Observe also that for $x\in[\tfrac{1}{2}, 1]$, 
Eqs.~(\ref{comp1}) and (\ref{comp2}) implies that  the maximum value  of $M'$ we can get  
for points $(x',M')$ of  ${\mathbb{H}}_{x,M}$ which have $x'\leq 1$ is
\begin{equation} \label{defm<} 
M^{(<)}_{\max}(x',x) : =\left\{ \begin{array}{ll} 
M+x'-x\;, & \forall x'\in [x-M,x],\\
\frac{x'}{x}M\;, & \forall x'\in [x,1].
\end{array} 
\right.
\end{equation} 
Therefore we can write
\begin{eqnarray}  
\underline{Q}^{(1)}_\mathrm{FKG}(x,M)&=& \min_{x'\in[x-M, 1]}
Q^{\mathrm{att}}_\mathrm{FKG}(x',\tfrac{M^{(<)}_{\max}(x',x)}{1-x'}) \\ &=&\nonumber
\min\{  \underline{Q}^{(1.1)}_\mathrm{FKG}(x,M), \underline{Q}^{(1.2)}_\mathrm{FKG}(x,M)\}\;,  
\end{eqnarray} 
with 
\begin{eqnarray} \underline{Q}^{(1.1)}_\mathrm{FKG}(x,M)&:=&
\min_{x'\in[x-M, x]}\label{new3eee}
Q^{\mathrm{att}}_\mathrm{FKG}(x',\tfrac{M+x'-x}{1-x'}) \\ \nonumber 
&=& \min_{\epsilon\in[0, 1]} Q^{\mathrm{att}}_\mathrm{FKG}(x-\epsilon M,\tfrac{(1-\epsilon) M}{1-x + \epsilon M}) \;, 
\end{eqnarray} 
where the second identify simply follows from a proper 
parametrization of  $x'$, and 
\begin{eqnarray} \label{new4} 
\underline{Q}^{(1.2)}_\mathrm{FKG}(x,M)
&:=& \min_{x'\in[x, 1]}
Q^{\mathrm{att}}_\mathrm{FKG}(x',\tfrac{x' M}{(1-x')x}) \\
&=& \min_{\epsilon\in[0, 1]}Q^{\mathrm{att}}_\mathrm{FKG}\left(1+ \epsilon(1-x),\tfrac{1+\epsilon (1-x) M}{(1-x)x\epsilon }\right)\;. \nonumber 
\end{eqnarray} 
Notice that for $\epsilon =0$ the function 
$Q^{\mathrm{att}}_\mathrm{FKG}(x-\epsilon M,\tfrac{(1-\epsilon) M}{1-x + \epsilon M})$  corresponds to $Q^{\mathrm{att}}_\mathrm{FKG}(x,M)$ in Eq.~(\ref{DEFFKBatt}) while  for $\epsilon=1$ we recover the bound $Q(\mathcal{E}_{x-M,0})$ of Eq.~(\ref{twist}). 
Therefore for the attenuators we have that $\underline{Q}^{(1.1)}_\mathrm{FKG}(x,M)$ (and hence $\underline{Q}^{(1)}_\mathrm{FKG}(x,M)$) 
is always guaranteed to provide bounds which are at least equal than those reported in Eqs.~(\ref{DEFFKBatt}) and~(\ref{twist}), i.e. 
\begin{equation}
    \underline{Q}^{(1)}_\mathrm{FKG}(x,M)\leq \min\{Q(\mathcal{E}_{x-M,0}),Q^{\mathrm{att}}_\mathrm{FKG}(x,M)\}\;.
\end{equation}
On the contrary for $x\geq 1$,
 Eq.~(\ref{defm<}) gets replaced 
\begin{equation} 
M^{(<)}_{\max}(x',x) : =
M +x'-1\;,  \forall x'\in [1-M,1],
\end{equation}
leading to 
\begin{eqnarray} \underline{Q}^{(1)}_\mathrm{FKG}(x,M)&:=&
\min_{x'\in[1-M, 1]}
Q^{\mathrm{att}}_\mathrm{FKG}(x',\tfrac{M+x'-1}{1-x'})\nonumber \\
&=& \min_{\epsilon\in[0, 1]}
Q^{\mathrm{att}}_\mathrm{FKG}\left(1-\epsilon M,\tfrac{1-\epsilon}{\epsilon}\right)\;. 
\end{eqnarray} 

\subsubsection{Numerical analysis}

Numerical analysis shows that $\underline{Q}^{(1.2)}_\mathrm{FKG}(x,M)$ is always less performant than $\underline{Q}^{(1.1)}_\mathrm{FKG}(x,M)$ so we can drop it form  Eq.~(\ref{asdf111}). 
Accordingly we can claim that for an attenuator channel ($x\leq 1$) the quantum capacity $Q(\Phi_{x,M})$ must fulfil two new sets of  inequalities, i.e.  
\begin{eqnarray} \label{newBound2ee} 
Q(\Phi_{x,M}) \leq  \underline{Q}^{(2)}_\mathrm{FKG}(x,M) = 
Q^{\mathrm{amp}}_\mathrm{FKG}(M/x) \;,
\end{eqnarray} 
which follows from (\ref{new2}) and (\ref{new3}),
and 
\begin{equation} 
\label{new1bounds} 
Q(\Phi_{x,M}) \leq  \underline{Q}^{(1)}_\mathrm{FKG}(x,M)= \min_{\epsilon \in [0,1]} 
 Q^{\mathrm{att}}_\mathrm{FKG}(x-\epsilon M,\tfrac{(1-\epsilon) M}{1-x + \epsilon M})\;, 
\end{equation}
which instead follows from (\ref{new3eee}).
As shown in   Fig.~\ref{fig:m=0.15}  for some value of the channel parameter these two functions provide better upper bounds that those reported in
 Sec.~\ref{sec:up}. \begin{figure}
	\centering
	\includegraphics[width=0.99\columnwidth]{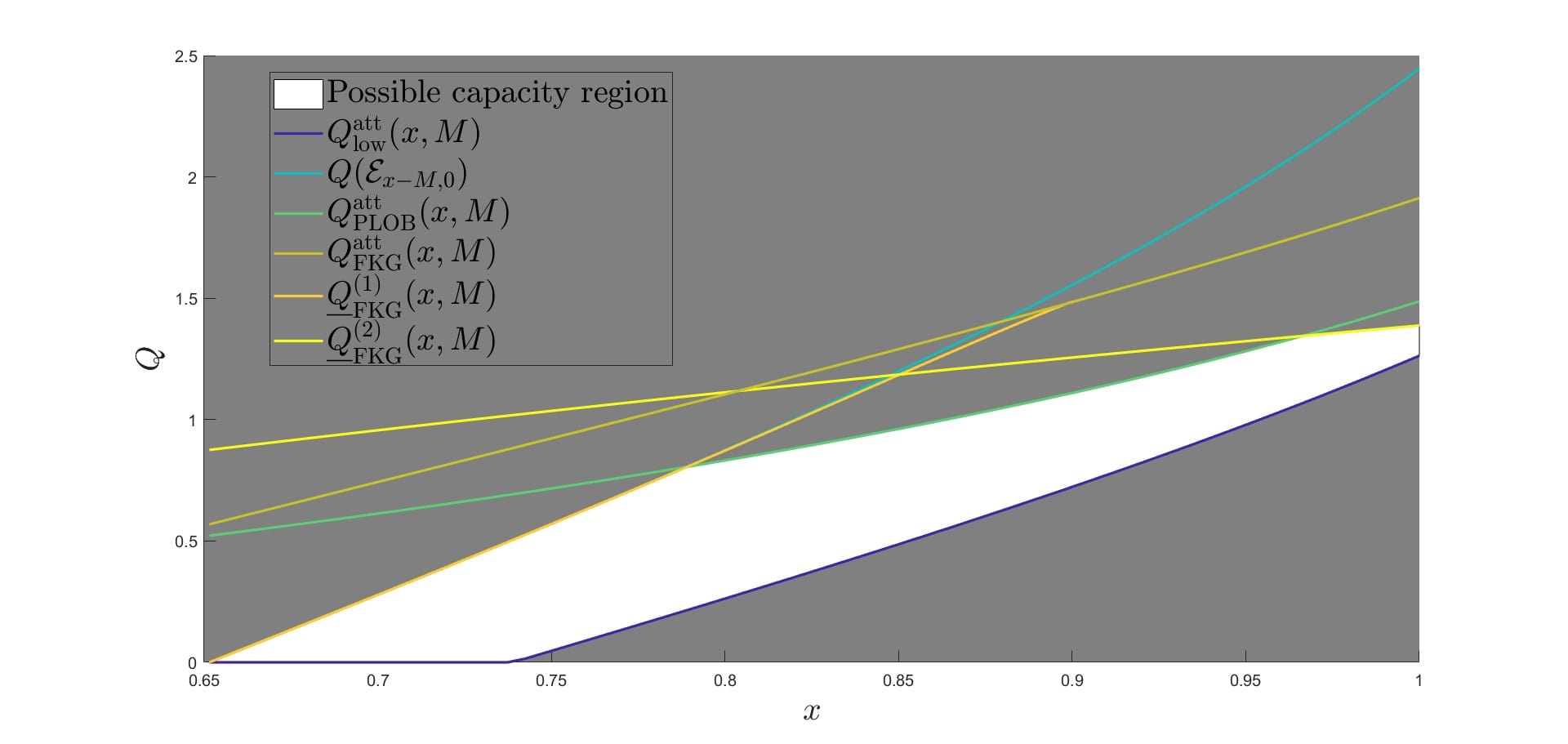}
	\caption{Comparison between state of the art upper bounds on the quantum capacity $Q(\Phi_{x,M})$ of thermal attenuator ($x\leq 1$) for $M=0.15$. Each color represents an uper  bound on the quantum capacity of thermal attenuator. For $x<0.78$, the function $\underline{Q}^{(1)}_\mathrm{FKG}(x,M)$ in Eq.~(\ref{new1bounds}) is the best upper bound (note that in this region $Q(\mathcal{E}_{\eta-N(1-\eta),0})$ and $\underline{Q}^{(1)}_\mathrm{FKG}(x,M)$ are very close, by numerical evidence). For $x>0.97$, the function $\underline{Q}^{(2)}_\mathrm{FKG}(x,M)$ in Eq.~(\ref{newBound2ee}) outperforms the other bounds, and the for intermediate values of $x$ i.e.  $0.78<x<0.97$, bound $Q^\mathrm{att}_\mathrm{PLOB}(x,M)$ in Eq.~(\ref{UPPLOB1}) wins. The purple line represents the lower bound $Q^\mathrm{att}_\mathrm{low}(x,M)$ in Eq.~(\ref{lower2}), and the white region indicates the possible values of quantum capacity of thermal attenuator.}
	\label{fig:m=0.15}
\end{figure}

For amplifiers (i.e. $x\geq 1$) we have already observed that  $\underline{Q}^{(2)}_\mathrm{FKG}(x,M)$ always coincides with the old bound (\ref{new00}).
Accordingly new results can only derive from  $\underline{Q}^{(1)}_\mathrm{FKG}(x,M)$ (i.e. from $\underline{Q}^{(1.1)}_\mathrm{FKG}(x,M)$): unfortunately 
numerical study reveals that such function is always less performant than the bounds of Sec.~\ref{sec:up}. 
A comprehensive comparison between the old bounds and the improved versions  derived in this section is presented in the right part of Fig.~\ref{fig:regions}.

\section{Conclusion}\label{conc}
 In the present manuscript, we sought to better understand the behavior of the capacities of single-mode phase-insensitive Gaussian bosonic channels (PI-GBCs) and their concatenation rules. Through extensive analysis, we were able to establish{, for each point in the parameter space of PI-GBCs, an analytical characterization of two regions in the parameter space, of higher and lower capacity.  That is,} the capacity of points within each region was found to be {respectively} higher or lower than that of the original channel. Using these regions, we were able to improve upon the previous upper bounds. This structure of parameter phase space of PI-GBCs can be used to potentially improve any new upper or lower bound.

We thank F. A. Mele and L. Lami for comments and suggestions. We also 
acknowledge financial support by MUR (Ministero dell’ Universit\`a e della Ricerca) through the PNRR MUR project PE0000023-NQSTI. MF is supported by Juan de la Cierva - Formación (Spanish MICIN project FJC2021-047404-I), with funding from MCIN/AEI/10.13039/501100011033 and European Union “NextGenerationEU”/PRTR.

	\let\oldaddcontentsline\addcontentsline
\renewcommand{\addcontentsline}[3]{}

\pagebreak
\newpage

\widetext
\appendix

\section{Monotonicity along the borders}  \label{sec:mono} 
Notice  that as the property~(\ref{primaprima}) imposes in particular that ${\cal K}(\mathcal{E}_{\eta,N})\geq {\cal K}(\mathcal{E}_{\eta',N'})$ 
for all  $\eta'\leq \eta$, $N'\geq N$, it follows that 
 ${\cal K}(\mathcal{E}_{\eta,N})$ 
must be non-decreasing  in $\eta$ for fixed $N$, and non-increasing w.r.t. $N$ for fixed $\eta$.
Similarly, from~(\ref{primaprimap4imq}) it follows that 
 that 
${\cal K}({\cal A}_{g,N})$ 
must be non-increasing  in $g$ for fixed $N$, and non-increasing w.r.t. $N$ for fixed $g$.

 More generally we can establish the following monotonicity rules for points that lies along the borders
 of the low-ground/high-ground regions, i.e. 
 
\begin{corollary}\label{monotonicity}
The capacity ${\cal K}$ of the channel $\mathcal{E}_{\eta',N'}$ as a function of $\eta'$ is monotonically non-decreasing when evaluated along points of the curve 
$N'={N}_{\eta,N}^{{(1)}}(\eta')$, and monotonically non-increasing when 
evaluated along points of the curve 
$N'={N}_{\eta,N}^{{(2)}}(\eta')$.  Analogously  the capacity 
${\cal K}$ of the channel ${\cal A}_{g',N'}$ as a function of $g'$ is monotonically non-increasing when evaluated along points of the curve 
$N'={N}_{g,N}^{{(1)}}(g')$, and monotonically non-decreasing when 
evaluated along points of the curve 
$N'={N}_{g,N}^{{(2)}}(g')$.
\end{corollary}
(It goes without mentioning that, in the case of 
${N}_{\eta,N}^{{(2)}}(\eta')$ and 
${N}_{g,N}^{{(2)}}(g')$ the above properties apply only for those
points where  these functions are explicitly non negative). 
\begin{proof}We report here only the proof for the attenuator, as the one for the amplifiers can be obtained by following the same steps. 
Let us start addressing the monotonicity along the curve $N'={N}_{\eta,N}^{{(1)}}(\eta')$.
If $\eta'_1\leq \eta'_2$  and $N'_1={N}_{\eta,N}^{{(1)}}(\eta'_1),\;N'_2={N}_{\eta,N}^{{(1)}}(\eta'_2)$, we can write 
\begin{eqnarray}
	N^{(1)}_{\eta'_2,N'_2}(\eta'_1)&=&(\tfrac{1-\eta'_2}{\eta'_2})N'_2(\tfrac{\eta'_1}{1-\eta'_1})=(\tfrac{1-\eta}{\eta})N(\tfrac{\eta'_1}{1-\eta'_1})\nonumber\\&=&{N}_{\eta,N}^{{(1)}}(\eta'_1)=N'_1\; .
\end{eqnarray}
Therefore we can conclude that $(\eta'_1,N'_1)\in  {\mathbb{L}}^{({\rm att})}_{\eta'_2,N'_2}$ which implies~${\cal K}( \mathcal{E}_{\eta'_1,N'_1}) \leq {\cal K}( \mathcal{E}_{\eta'_2,N'_2})$  as a direct consequence of (\ref{IMPO12}) of Property~\ref{th3}. 
Consider next the  monotonicity along the curve $N'={N}_{\eta,N}^{{(2)}}(\eta')$ for points $\eta'$ for which ${N}_{\eta,N}^{{(2)}}(\eta')$ is strictly positive. 
Given then  $\eta'_1\leq \eta'_2$ values that fulfils such constraint  can write 
\begin{eqnarray}
	N^{(2)}_{\eta'_1,N'_1}(\eta'_2)&=& (N'_1+1)\left(\tfrac{1-\eta'_1}{1-\eta'_2}\right)-1 \nonumber
	\\&=& (N+1)\left(\tfrac{1-\eta}{1-\eta'_2}\right)-1 \nonumber \\
	&=& {N}_{\eta,N}^{{(2)}}(\eta'_2)=N'_2\; . \label{derivazione1} 
\end{eqnarray}
 The identity~(\ref{derivazione1}) implies that 
 $(\eta'_2,N'_2)\in  {\mathbb{L}}^{(\rm att)}_{\eta'_1,N'_1}$ that finally 
leads to 
${\cal K}( \mathcal{E}_{\eta'_1,N'_1}) \geq {\cal K}( \mathcal{E}_{\eta'_2,N'_2})$.
\end{proof}
The above results also allow us to establish monotonicity rules for  the curves~${N}_{\eta,N}^{(3,4)}(g)$ and ${N}_{g,N}^{(3,4)}(\eta)$:
\begin{corollary}\label{monotonicitynew1}
The capacity ${\cal K}$ of the channel $\mathcal{E}_{\eta',N'}$ as a function of $\eta'$ is monotonically non-decreasing when evaluated along points of the curve $N'={N}_{g,N}^{(4)}(\eta')$, and monotonically non-increasing when 
evaluated along points of the curve 
 $N'={N}_{g,N}^{(3)}(\eta')$. Analogously  the capacity 
${\cal K}$ of the channel ${\cal A}_{g',N'}$ as a function of $g'$ is monotonically non-increasing when evaluated along points of the curve 
$N'={N}_{\eta,N}^{(4)}(g')$, and monotonically non-decreasing when 
evaluated along points of the curve 
$N'={N}_{\eta,N}^{(3)}(g')$.\end{corollary}
\begin{proof}
The derivation relies on the fact that we can map the curves ${N}_{g,N}^{(3,4)}(\eta')$ and ${N}_{\eta,N}^{(3,4)}(g')$ into 
  ${N}_{\eta,N}^{(1,2)}(\eta')$ and ${N}_{g,N}^{(1,2)}(g')$ respectively. We give here direct proof of this fact  only for 
  ${N}_{\eta,N}^{(4)}(\eta')$: the generalization to the other cases follows trivially. 
 Consider a point $(\eta_1,N_1)$ on ${N}_{g,N}^{(4)}(\eta')$: inverting the identity 
  $N_1= {N}_{g,N}^{(4)}(\eta_1)$ can write 
  \begin{eqnarray} (N+1)
		\left(\tfrac{g-1}{g}\right)= N_1 \left(\tfrac{1-
		\eta_1}{\eta_1}\right)\;,
		\end{eqnarray} 
		which gives
  \begin{eqnarray} {N}_{g,N}^{(4)}(\eta') &=&  (N+1)
		\left(\tfrac{g-1}{g}\right)\left(\tfrac{
		\eta'}{1-\eta'}\right) \nonumber \\
		&=& N_1 \left(\tfrac{1-
		\eta_1}{\eta_1}\right)\left(\tfrac{
		\eta'}{1-\eta'}\right)  =  {N}_{\eta_1,N_1}^{(1)}(\eta')\;. 
		\end{eqnarray} 	
		Invoking hence Corollary~\ref{monotonicity} we can now claim that ${\cal K}(\mathcal{E}_{\eta',N'})$ will be
		non-decreasing w.r.t. to $\eta'$, hence proving the thesis. 
\end{proof}

\section{Formal derivation of Eq.~(\ref{ddfs})} \label{app:formaL} 
The derivation of  Eq.~(\ref{ddfs}) presented in the main text suffers from
a problem related with the fact that, technically speaking,  Eq.~(\ref{notalimite}) is true only in an approximate sense.
Indeed from~(\ref{defNnew}) it is clear that any composition involving the channel $\Phi_{0,0}$ is bound to produce only 
elements with $x=0$, i.e. 
\begin{eqnarray} 
\Phi_{\bar{x}_2,\bar{M}_2} \circ \Phi_{0,0} \circ \Phi_{\bar{x}_1,\bar{M}_1} = 
 \Phi_{0,\bar{M}_2+\Theta(\bar{x}_2-1)}\;. 
\end{eqnarray} 
Therefore as soon as $x>0$ there is no hope to see $(0,0)$ as an element of its high-ground set  ${\mathbb{H}}_{x,M}$.
The way one should interpret  Eq.~(\ref{notalimite}) is that for all   $(x,M)$ such that (\ref{EBconstraint}) 
holds true then we can a set of points arbitrarily close  to $(0,0)$ which 
belong to  ${\mathbb{H}}^{(0)}_{x,M}$ and hence into ${\mathbb{H}}_{x,M}$. To see that this is enough to prove 
Eq.~(\ref{ddfs}) observe that according to Eq.~(\ref{comp3}) and (\ref{comp4}), for $x'\in ]0,\min\{x,1\}]$ the conditions 
under which $(x',M')$ can be included into ${\mathbb{H}}^{(0)}_{x,M}$ (and hence into ${\mathbb{H}}_{x,M}$) can be expressed as
\begin{eqnarray}  \label{uno0} 
0 \leq M' \leq M+x' - \min\{x,1\}\;, 
\end{eqnarray} 
 a region which is not empty if (\ref{EBconstraint}) holds true.
Observe also that according to Eqs.~(\ref{comp1}) and (\ref{comp2}), for all $(x_1,M_1)\in \left(\mathbb{R}^{+} \right)^{2}$, the  
points $(x',M')$ such that 
\begin{eqnarray} \label{uno1} 
 M' \geq \left( M_1+(x_1-1)\Theta(x_1-1)\right)x'/x_1\;, \qquad x' \in [0,\min\{x_1,1\}]\;,
\end{eqnarray} 
 are always included into ${\mathbb{L}}^{(0)}_{x_1,M_1}$ (and hence 
into ${\mathbb{L}}_{x_1,M_1}$). 
Taking then $(x,M)$ EB with $x > M_{\text{EB}}(x)$, and $(x_1,M_1)$ generic, we notice that 
for all $x'' \in ]0,\min\{x_1,1\}]$ and ${M}'':= \left( M_1+(x-1)\Theta(x_1-1)\right)x''/x$ the point 
$(x',\bar{M}')$ fulfils both Eq.~(\ref{uno0}) and (\ref{uno1}). Therefore we can write  
\begin{eqnarray}\left.
\begin{array}{l}(x'',{M}'') \in{\mathbb{H}}_{x,M} \;, \\ 
(x'',{M}'') \in {\mathbb{L}}_{x_1,M_1} \Longrightarrow 
(x_1,{M}_1) \in {\mathbb{H}}_{x'',{M}''} \;,
\end{array} \right\} \Longrightarrow (x_1,{M}_1) \in {\mathbb{H}}_{x,{M}}\;,
\end{eqnarray} 
where the first implication is a consequence of the complementary relation~(\ref{complementary}), and the second of the 
natural ordering~(\ref{Hide11}).


\begin{thebibliography}{99}
\bibitem{shannon1}	C. E. A Shannon,  Bell Syst. Tech. J. {\textbf{27}}, 379–423 (1948)

\bibitem{shannon2}	C. E. A Shannon,  Bell Syst. Tech. J. {\textbf{27}}, 623–656 (1948)


\bibitem{BENSHOR} 
C. H. Bennett and P. W. Shor,  IEEE Trans. Inf. Th.  {\bf 44}, 2724-2742, (1998).
 \bibitem{VGHOL} 
V. Giovannetti and H.S. Holevo Rep. Prog. Phys. {\bf 75} 046001 (2012). 

\bibitem{HOL BOOK}A. S. Holevo, {\it Quantum systems, channels, information: a mathematical introduction}  (de Gruyter, 2012).	
\bibitem{WILDE BOOK} M. Wilde, {\it Quantum Information Theory} (Cambridge University Press, 2013).

\bibitem{CAVES} C. M. Caves and P. D. Drummond
Rev. Mod. Phys. {\bf 66}, 481 (1994). 

\bibitem{HOLEVO01} A. S. Holevo and R. F. Werner, Phys. Rev. A, {\bf 63} 032312, (2001).

\bibitem{serafini} A. Serafini, {\it Quantum Continuous Variables: A Primer
	of Theoretical Methods} (CRC Press, London, 2017).
	
	\bibitem{ADVS}
	J. S. Sidhu {\it et al.} IET Quantum Communication, {\bf 2}:182, (2021).
\bibitem{S} P. Shor, {\it The quantum channel capacity and coherent information}. Lecture notes, MSRI Workshop
on Quantum Computation  (2002).







\bibitem{Lioyd} S. Lloyd, 
Phys. Rev. A {\bf 55},  1613 (1997).





\bibitem{Q CAP DEV}I. Devetak, IEEE Trans. Inf. Th. {\bf 51}, 44 (2005).


\bibitem{privCWY} N.~Cai, A.~Winter, and R.~W. Yeung, Problems of Information Transmission {{\bf40}, 318  (2004)}. 
{\bibitem{privDS} I.~Devetak and P.~W. Shor, Commun. Math. Phys. {\bf256}, 287 (2015)
}


\bibitem{shor superadditivity DC}
P. W. Shor, and J. A. Smolin, arXiv preprint quant-ph/9604006 (1996).


\bibitem{Di vincenzo superadditivity DC} 
D. P. DiVincenzo, P. W. Shor, and J. A. Smolin, Phys. Rev. A {\bf 57}, 830 (1998).

\bibitem{smith superadditivity} 
G. Smith and J. A. Smolin, Phys. Rev. Lett. {\bf 98},  030501 (2007).



\bibitem{Fern superadditivity}
J. Fern and K. B. Whaley, Phys. Rev. A {\bf 78}, 062335 (2008).






\bibitem{q superadditivity}
G. Smith and J. Yard. Science {\bf 321}, 1812 (2008).



\bibitem{gaussian q superadditivity}
G. Smith, J. Smolin, and J. Yard. Nat. Phot. {\bf 5}, 624 (2011).

\bibitem{Cubit superadditivity}
T. Cubitt, D. Elkouss, W. Matthews, M. Ozols, D. P\'erez-Garc\'{\i}a, and S. Strelchuk, Nat. Comm. {\bf 6}, 1 (2015).



\bibitem{dephrasure} F. Leditzky, D. Leung, and G. Smith, Phys. Rev. Lett. {\textbf{121}},  160501 (2018).	

\bibitem{bl superadditivity}
{J. Bausch and F. Leditzky, SIAM J. Comput. {\bf 50} 1410 (2021)
}

\bibitem{c superadditivity}
M. Hastings, Nat. Phys. {\bf 5}, 255 (2009).

\bibitem{p superadditivity}
K. Li, A. Winter, X. B. Zou, and G. C. Guo. Phys. Rev. Lett. {\bf 103}, 120501 (2009).

\bibitem{tradeoff1}E. Y. Zhu, Q. Zhang and P. W. Shor, Phys. Rev. Lett. {\bf 119}, 040503 (2017).

\bibitem{tradeoff2} E. Y. Zhu, Q. Zhang, M-H Hsieh, and P. W. Shor, IEEE Trans.  Inf. Th. {\bf 65}, 3973 (2018).

{\bibitem{vikesh} V. Siddhu,  Nat. Comm. {\bf 12}, 5750 (2021).
}







\bibitem{SS UPP B DEP}G. Smith  and J. A. Smolin, IEEE {\it Information Theory Workshop}, vol. {\bf 54}, 4208  (2008).


\bibitem{Ouyang}
Y. Ouyang, Quantum Information \& Computation {\bf 14}, 917 (2014).



\bibitem{SUTT UPP B DEP}D. Sutter, V. B. Scholz, A. Winter, and R. Renner, IEEE Trans. Inf. Th.  {\bf 63}, 7832 (2017).


\bibitem{lownoiseQ} F. Leditzky, D. Leung, and G. Smith, Phys. Rev. Lett. {\textbf{120}}, 160503 (2018).

\bibitem{distdepo}F. Leditzky, N. Datta, and G. Smith, 
IEEE Trans. Inf. Th. {\bf 64}  4689 (2018).



\bibitem{flagged channel 1}
M. Fanizza, F. Kianvash, and V. Giovannetti, Phys. Rev. Lett. {\textbf{125}},  020503  (2020).

{\bibitem{flagged channel 2}
F. Kianvash, M. Fanizza, and V. Giovannetti, Quantum {\bf 6}, 647 (2022)
}
\bibitem{wang}
{X. Wang, IEEE Trans. Inf. Th. {\bf 67}, 4524 (2021).
}

\bibitem{swat}
K. Sharma, M. M. Wilde, S. Adhikari \& M. Takeoka, New J. Phys. {\bf20}, 063025 (2018).



{\bibitem{naj}
K. Noh, V.V. Albert, and L. Jiang, IEEE Trans. Inf. Th. {\bf 65}, 2563 (2018).
}




\bibitem{matteo}
M. Rosati, A. Mari, and V Giovannetti, Nat. Comm. {\bf 9}, 4339 (2018)

\bibitem{flagged channel 3}
M. Fanizza, F. Kianvash, and V. Giovannetti, Phys. Rev. Lett. {\textbf{127}},  210501   (2021).

\bibitem{Pirandola upp bound}
S. Pirandola, R. Laurenza, C. Ottaviani, and L Banchi, Nat. Comm. {\bf 8}, 15043 (2017).

\bibitem{degradability} I. Devetak and P. W. Shor, Comm. Math. Phys. \textbf{256}, 2 (2005).


\bibitem{BENNETT1} C. H. Bennett, P. W. Shor, J. A. Smolin, and A. V. Thpliyal, Phys. Rev. Lett. {\bf 83}, 3081 (1999). 

	\bibitem{MELE} F. A. Mele, L. Lami, and V. Giovannetti, arXiv:2303.12867 [quant-ph].
	
	\bibitem{LAMI19} L. Lami, S. Khatri, G. Adesso, and M. M. Wilde, Phys. Rev. Lett. {\bf 123}, 050501 (2019).
\bibitem{CARUSO061} F Caruso, and V. Giovannetti, Physical Review A {\bf 74}, 062307 (2006).
\bibitem{CARUSO06} F Caruso, V. Giovannetti, and A S Holevo, New Journal of Physics {\bf 8}, 310 (2006).
\bibitem{HOLEVO07} A. S. Holevo, Problems of Information Transmission {\bf 43}, 1 (2007).

\bibitem{WOLF} M. M. Wolf, D. Perez-Garcia, G. Giedke, Phys. Rev. Lett. 98, 130501 (2007).


\bibitem{btw}
M.M. Wilde, M. Tomamichel, and M. Berta,  IEEE Trans. Inf. Th. {\bf 63}, {1792} (2017).




\bibitem{PIR09} 
S. Pirandola, R. Garcia-Patron, S. L. Braunstein, and S. Lloyd, Phys. Rev. Let. {\bf 102} 050503 (2009).
\bibitem{PIR16} 
C. Ottaviani, R. Laurenza, T.P. Cope, G. Spedalieri, S.L. Braunstein, and S. Pirandola, October. Quantum Information Science and Technology II SPIE \textbf{9966}, 16 (2016)


	
	\bibitem{NOTA1} Notice that, thanks to~(\ref{defNnew0}), 
	 the definition~(\ref{alternativedefLO})  includes also all those cases where $\Phi_{x',M'}$ and 
	$\Phi_{x,M}$ are connected with 
	more complex decompositions that  involve more than one   left-most  or on the right-most elements. For instance 
$\Phi_{x',M'} =\left(\Phi_{\bar{x}_{1,1},\bar{M}_{1,1}} \circ\Phi_{\bar{x}_{1,2},\bar{M}_{1,2}} \right)\circ  \Phi_{x,M} \circ \Phi_{\bar{x}_2,\bar{M}_2}$ can be casted into the form  $\Phi_{x',M'} =\Phi_{\bar{x}_1,\bar{M}_1} \circ  \Phi_{x,M} \circ \Phi_{\bar{x}_2,\bar{M}_2}$ by simply identifying $\Phi_{\bar{x}_1,\bar{M}_1}$ with $\Phi_{\bar{x}_{1,1},\bar{M}_{1,1}} \circ\Phi_{\bar{x}_{1,2},\bar{M}_{1,2}}$. Similar considerations of course apply also to Eq.~(\ref{alternativedefHIGH}). 
\bibitem{nota0} 
In writing~(\ref{analytical2new}) we used the fact that 
that $f^{(1)}_{0,0}(x')=\min\{  1,x'\}$, and that the proper limit
$x\rightarrow 0, M\rightarrow 0$ of $f^{(2)}_{x,M}(x')$ is given by $f^{(2)}_{0,0}(x')=- (x'-1)\Theta(x'-1)$.

	
	
	
\end{thebibliography}
\end{document}